\documentclass[11pt]{article} 
\usepackage{amsfonts,amssymb,amsmath,amsthm,graphicx,hyperref,mathrsfs,bm,bbm}
\usepackage[title, titletoc]{appendix}
\usepackage[onehalfspacing]{setspace}
\usepackage{caption,subcaption,lscape}
\hypersetup{pdfborder = {0 0 0},colorlinks=true,linkcolor=blue,urlcolor=blue,citecolor=blue}
\usepackage{natbib}
\usepackage[shortlabels]{enumitem}
\usepackage{thmtools}
\renewcommand\thmcontinues[1]{Continued}

\usepackage{xcolor}
\usepackage{colortbl}
\usepackage{booktabs}

\usepackage{tikz}
\usepackage{pgfplots}
\pgfplotsset{compat=1.18}

\usepackage[top=1in,bottom=1in,left=1in,right=1in]{geometry}

\newtheorem{thm}{Theorem}
\newtheorem{coro}{Corollary}
\newtheorem{lem}{Lemma}
\newtheorem{prop}{Proposition}
\theoremstyle{definition}

\newtheorem{assumption}{Assumption}
\newtheorem*{assumption*}{Assumption}

\numberwithin{equation}{section}

\theoremstyle{definition}
\newtheorem{remark_tmp}{Remark}[section]

\allowdisplaybreaks[1]

\renewcommand{\arraystretch}{.6}

\newcommand\independent{\protect\mathpalette{\protect\independenT}{\perp}}
\def\independenT#1#2{\mathrel{\rlap{$#1#2$}\mkern2mu{#1#2}}}
\DeclareMathOperator*{\argmin}{arg\,min}

\DeclareMathOperator*{\sgn}{sgn}
		
\renewcommand{\P}{\mathbb{P}}
\newcommand{\E}{\mathbb{E}}
\newcommand{\V}{\mathbb{V}}
\newcommand{\R}{\mathbb{R}}

\newcommand{\Z}{\mathbb{Z}}
\newcommand{\1}{\mathbbm{1}}

\def\b{\boldsymbol}

\newcommand{\0}{\mathbf{0}}

\newcommand{\ba}{\mathbf{a}}
\newcommand{\bb}{\mathbf{b}}
\newcommand{\bc}{\mathbf{c}}

\newcommand{\bs}{\mathbf{s}}
\newcommand{\bt}{\mathbf{t}}
\newcommand{\bu}{\mathbf{u}}

\newcommand{\bx}{\mathbf{x}}
\newcommand{\bz}{\mathbf{z}}
\newcommand{\bw}{\mathbf{w}}

\newcommand{\bG}{\mathbf{G}}
\newcommand{\bH}{\mathbf{H}}
\newcommand{\bI}{\mathbf{I}}

\newcommand{\bL}{\mathbf{L}}

\newcommand{\bU}{\mathbf{U}}
\newcommand{\bV}{\mathbf{V}}
\newcommand{\bW}{\mathbf{W}}

\newcommand{\bbeta}{\boldsymbol{\beta}}

\newcommand{\bmu}{\boldsymbol{\mu}}
\newcommand{\bnu}{\boldsymbol{\nu}}
\newcommand{\bvarphi}{\boldsymbol{\varphi}}
\newcommand{\bpsi}{\boldsymbol{\psi}}
\newcommand{\btheta}{\boldsymbol{\theta}}

\newcommand{\bomega}{\boldsymbol{\omega}}
\newcommand{\bxi}{\boldsymbol{\xi}}

\newcommand{\bell}{\boldsymbol{\ell}}

\newcommand{\bDelta}{\boldsymbol{\Delta}}
\newcommand{\bGamma}{\boldsymbol{\Gamma}}
\newcommand{\bSigma}{\boldsymbol{\Sigma}}

\newcommand{\bXi}{\boldsymbol{\Xi}}

\newcommand{\cN}{\mathcal{N}}
\newcommand{\cW}{\mathcal{W}}

\newcommand{\smooth}{\mathtt{S}}

\begin{document}

% % % % % % % % % % % % % % Title Page % % % % % % % % % % % % % %
\hypersetup{pageanchor=false}
\title{\vspace{-0.0in} Robust Inference for Convex Pairwise Difference Estimators\thanks{This paper was prepared for the Econometric Theory Lecture delivered at the $2025$ International Symposium on Econometric Theory and Applications (SETA), University of Macau (China), June 1--3, 2025. It was also presented at the Econometrics Journal Lecture of the $2024$ $(\text{EC})^2$ Conference (Amsterdam), and the $2025$ Conference in Honor of Bo Honor{\'e}'s $65$th Birthday (Princeton University). We thank the participants at these conferences for their feedback. Cattaneo gratefully acknowledges financial support from the National Science Foundation through grants SES-1947805, DMS-2210561, and SES-2241575. Jansson gratefully acknowledges financial support from the National Science Foundation through grant SES-1947662 and from the Aarhus Center for Econometrics (ACE) funded by the Danish National Research Foundation grant number DNRF186. Nagasawa gratefully acknowledges financial support from the British Academy through grant SRG24$\backslash$241614.}
\bigskip }

\author{Matias D. Cattaneo\thanks{Department of Operations Research and Financial Engineering, Princeton University.} \and
	Michael Jansson\thanks{Department of Economics, UC Berkeley and ACE.} \and
	Kenichi Nagasawa\thanks{Department of Economics, University of Warwick.}}
\maketitle

\begin{abstract}
    This paper develops distribution theory and bootstrap-based inference methods for a broad class of convex pairwise difference estimators. These estimators minimize a kernel-weighted convex-in-parameter function over observation pairs with similar covariates, where the similarity is governed by a localization (bandwidth) parameter. While classical results establish asymptotic normality under restrictive bandwidth conditions, we show that valid Gaussian and bootstrap-based inference remains possible under substantially weaker assumptions. First, we extend the theory of small bandwidth asymptotics to convex pairwise difference estimation settings, deriving robust Gaussian approximations even when a smaller than standard bandwidth is used. Second, we employ a debiasing procedure based on generalized jackknifing to enable inference with larger bandwidths, while preserving convexity of the objective function. Third, we construct a novel bootstrap method that adjusts for bandwidth-induced variance distortions, yielding valid inference across a wide range of bandwidth choices. Our proposed inference method enjoys demonstrably greater robustness, while retaining the practical appeal of convex pairwise difference estimators.
\end{abstract}


\thispagestyle{empty}

%\clearpage
%\tableofcontents

\thispagestyle{empty}
\clearpage

\hypersetup{pageanchor=true}
\setcounter{page}{1}
\pagestyle{plain}

\section{Introduction}

Suppose $\bz_1,\dots,\bz_n$ is a random sample from the distribution of a random vector $\bz$. This paper studies the large-sample properties of the following \textit{convex} pairwise difference estimator:
\begin{equation}\label{eq: Pairwise Difference Estimator}
    \widehat{\btheta}_n \in \argmin_{\btheta\in\Theta} \binom{n}{2}^{-1} \sum_{i<j} m(\bz_i,\bz_j;\btheta) K_{h_n}(\bw_i-\bw_j), \qquad K_h(\bu)=\frac{1}{h^d}K\left(\frac{\bu}{h}\right),
\end{equation}
where $\Theta \subseteq \R^k$ is a parameter space, $\sum_{i<j}$ denotes $\sum_{j=2}^n \sum_{i=1}^{j-1}$, $(\bz,\bar{\bz}) \mapsto m(\bz,\bar{\bz};\btheta)$ is a permutation symmetric function, $K$ is a symmetric, non-negative kernel, $h_n$ is a positive bandwidth (or localization) parameter sequence, $\bw$ is a continuously distributed $d$-dimensional subvector of $\bz$, and where $\btheta \mapsto m(\bz_i, \bz_j; \btheta)$ is a \textit{convex} function. Pairwise difference estimation, which relies on local comparisons between observation pairs, has been used to address heterogeneity in nonlinear models. See \citet{Powell_1994_Handbook}, \citet{Honore-Powell_2005_Festschrift}, and \citet*{AradillasLopez-Honore-Powell_2007_IER} for overviews, and Section \ref{Section: Motivating Examples} for three motivating examples.

In contrast to classical extremum estimators, $\widehat{\btheta}_n$ is a local $M$-estimator that employs observation pairs $(i, j)$ for which $\bw_i$ and $\bw_j$ are similar. The bandwidth $h_n$ governs the degree of similarity: When $h_n \to 0$ (as $n \to \infty$), the estimator increasingly focuses on nearly-identical-in-$\bw$ pairs. In turn, focusing on such pairs is natural in settings where identification can be based on the condition $\bw_i \approx \bw_j$ (combined with smoothness assumptions). The localization introduces a familiar trade-off for estimation and inference: A smaller $h_n$ reduces bias from dissimilarity between $\bw_i$ and $\bw_j$, but increases variance due to fewer available usable pairs. As a consequence, the large-sample behavior of $\widehat{\btheta}_n$ depends critically on a delicate bias-variance trade-off determined by $h_n$. This paper develops novel inference methods for convex pairwise difference estimators that are demonstrably more robust to bandwidth choice than existing methods.

Under regularity conditions and assuming that
\begin{equation}\label{eq: Bandwidth Conditions}
    nh_n^d \to \infty \qquad\text{and}\qquad nh_n^4 \to 0,
\end{equation}
the pairwise difference estimator $\widehat{\btheta}_n$ is known to be asymptotically linear:
\begin{equation}\label{eq: Asymptotic Linearity}
    \sqrt{n} (\widehat{\btheta}_n - \btheta_0) = \frac{1}{\sqrt{n}} \sum_{i=1}^n \bpsi_0(\bz_i) + o_\P(1) \rightsquigarrow \mathsf{N}(\0, \E[\bpsi_0(\bz)\bpsi_0(\bz)']),
\end{equation}
where $\btheta_0$ is the estimand, $\bpsi_0(\cdot)$ is an influence function (whose exact form is given below), and where $\rightsquigarrow$ denotes weak convergence. Moreover, the nonparametric bootstrap approximation to the distribution in \eqref{eq: Asymptotic Linearity} is consistent in the sense that
\begin{equation}\label{eq: Asymptotic Linearity - Bootstrap}
    \sqrt{n} (\widehat{\btheta}_n^* - \widehat{\btheta}_n) \rightsquigarrow_\P \mathsf{N}(\0, \E[\bpsi_0(\bz)\bpsi_0(\bz)']),
\end{equation}
where $\widehat{\btheta}_n^*$ is the bootstrap analogue of $\widehat{\btheta}_n$ and where $\rightsquigarrow_\P$ denotes weak convergence in probability.

%Here, the condition $nh_n^d \to \infty$ in \eqref{eq: Bandwidth Conditions} lower bounds the level of localization $h_n$ allowed for, while the condition $nh_n^4 \to 0$ upper bounds the level of localization. 
%The purpose of the latter condition is to control a smoothing bias term. The bias condition $nh_n^{4} \to 0$ could be replaced by the weaker condition $nh_n^{2L} \to 0$ if a (higher-order) kernel of order $L>2$ were used, but a higher-order kernel annihilates the convexity of the objective function because higher-order kernels take negative values.

The main results of this paper generalize \eqref{eq: Bandwidth Conditions}-\eqref{eq: Asymptotic Linearity - Bootstrap} by combining three ideas:
\begin{enumerate}
    \item \textit{Small Bandwidth Asymptotics}. Utilizing the framework introduced by \cite*{Cattaneo-Crump-Jansson_2014a_ET}, we obtain a Gaussian distributional approximation for the pairwise difference estimator without imposing the condition $nh_n^d \to \infty$, thereby allowing for higher levels of localization. This generalized distributional approximation shows in particular that, while the localization restriction $nh_n^d \to \infty$ is necessary for establishing asymptotic linearity, a valid Gaussian approximation can be obtained under the substantially weaker condition $n^2h_n^d\to\infty$, albeit with a convergence rate (and approximate variance) that depends explicitly on the level of localization used.
    
    \item \textit{Debiasing}. Following \cite{Honore-Powell_2005_Festschrift}, we debias the pairwise difference estimator using the method of \textit{generalized jackknifing} introduced by \cite{Schucany-Sommers_1977_JASA}. Doing so allows for (larger) bandwidths that violate the bias condition $nh_n^4 \to 0$. The same goal could be achieved by replacing the (second-order) kernel $K$ with a higher-order kernel, but a higher-order kernel annihilates the convexity of the objective function because higher-order kernels take negative values. In contrast, generalized jackknifing retains the convexity of objective functions, which in turn is attractive for both theoretical (weaker regularity conditions) and practical (faster computation) reasons.
    % The debiasing procedure combines linearly a collection of convex pairwise difference estimators constructed using different levels of localization. The resulting ensembling-based pairwise difference estimator admits a small bandwidth Gaussian approximation with an associated bias condition of the form $nh_n^{2L} \to 0$, where $L \geq 2$ denotes the order of a certain (equivalent) kernel induced by the debiasing procedure.

    \item \textit{Bootstrapping}. Building on insights from \citet*{Cattaneo-Crump-Jansson_2014b_ET}, we develop a valid bootstrap-based distributional approximation for the debiased pairwise difference estimator by rescaling the localization parameter. The nonparametric bootstrap distributional approximation exhibits a mismatch in its asymptotic variance under small bandwidth asymptotics. The mismatch is characterized by a known multiplicative factor involving the localization parameter $h_n$. As a result, bootstrapping the (debiased) pairwise difference estimator with a different localization parameter (namely, $3^{1/d}h_n$ rather than $h_n$) leads to a valid bootstrap-based inference procedure also under small bandwidth asymptotics.
\end{enumerate}

In combination, these three ideas enable us to offer a novel resampling-based inference method for (convex) pairwise difference estimators that is demonstrably more robust to the choice of the localization parameter $h_n$ than methods based on \eqref{eq: Bandwidth Conditions}-\eqref{eq: Asymptotic Linearity - Bootstrap}.

Our theoretical work is carefully developed to retain and leverage convexity of the objective function defining the pairwise difference estimator. This feature not only allows for fast implementation of the estimator and resampling-based methods, but also enables us to proceed under relatively weak conditions when obtaining theoretical results. When developing our theoretical results, we rely heavily on the foundational work of \cite{Hjort-Pollard_1993} and \cite{Pollard_1991_ET}, which we apply to the case of $U$-processes.

This paper is connected to several strands of the literature. Contributions to the pairwise difference estimation literature include \cite*{Ahn-Ichimura-Powell-Ruud_2018_JBES}, \cite{Ahn-Powell_1993_JoE}, \cite{AradillasLopez_2012_JoE}, \cite{Blundell-Powell_2004_REStud}, \cite{Hong-Shum_2010_REStud}, \cite{Honore_1992_ECMA}, \cite*{Honore-Kyriazidou-Udry_1997_JoE}, \cite{Honore-Powell_1994_JoE}, \cite{Jochmans_2013_EctJ}, \cite{Kyriazidou_1997_ECMA}, and \cite{Powell_2001_Festschrift}. The theoretical and practical features of small bandwidth asymptotics, and their connection with resampling methods for inference, are discussed in \citet*{Cattaneo-Crump-Jansson_2010_JASA,Cattaneo-Crump-Jansson_2014a_ET,Cattaneo-Crump-Jansson_2014b_ET}, \citet*{Cattaneo-Farrell-Jansson-Masini_2025_JoE}, \cite{Cattaneo-Jansson_2018_ECMA,Cattaneo-Jansson_2022_ET}, \cite*{Cattaneo-Jansson-Newey_2018_ET}, \cite{Matsushita-Otsu_2021_Biometrika}, and references therein. The generalized jackknife has been successfully used for debiasing in density weighted average derivative estimation \citep*{Powell-Stock-Stoker_1989_ECMA}, asymptotically linear pairwise difference estimation \citep{Honore-Powell_2005_Festschrift}, nonlinear semiparametric estimation \citep*{Cattaneo-Crump-Jansson_2013_JASA}, monotone estimation \citep*{Cattaneo-Jansson-Nagasawa_2024_AoS}, and random forest estimation \citep*{Cattaneo-Klusowski-Underwood_2026_JRSSB}, among other settings. \cite{Shao-Tu_2012_Book} give a textbook introduction to jackknifing, bootstrapping, and other resampling methods. 

The rest of the paper proceeds as follows. Section \ref{Section: Motivating Examples} introduces the three examples that are used throughout the paper to motivate our work and to illustrate the verification of the high-level assumptions imposed. Section \ref{Section: Distributional Approximation and Bootstrap Inference} presents our main results. The proofs of these results are given in Section \ref{Section: Proofs and Other Technical Results}. Section \ref{Section: Sufficient Conditions for Motivating Examples} revisits the three motivating examples and gives primitive conditions under which these examples are covered by our general theory. Simulation evidence is reported in Section \ref{Section: Simulation Evidence}. Section \ref{Section: Conclusion} gives final remarks.

\section{Motivating Examples}\label{Section: Motivating Examples}

We use three examples to motivate and illustrate our work. The first example involves an estimator that can be written in closed form (because it has a quadratic-in-$\btheta$ function $m(\bz_i, \bz_j; \btheta)$), while the other two examples do not. The second example has a smooth-in-$\btheta$ function $m(\bz_i, \bz_j; \btheta)$, while the third example does not. All three examples have convex-in-$\btheta$ functions $m(\bz_i, \bz_j; \btheta)$ and employ the following notation: $\bz_i = (y_i, \bx_i', \bw_i')'$, with $y_i$ a scalar outcome variable, $\bx_i$ a $k$-dimensional covariate, and $\bw_i$ a $d$-dimensional covariate. For more details on the examples, see \cite{Powell_1994_Handbook}, \cite{Honore-Powell_2005_Festschrift}, and \citet{AradillasLopez-Honore-Powell_2007_IER}.

\subsection{Partially Linear Regression Model}

The partially linear regression model is of the form
\begin{equation*}
    y_i =\bx_i'\btheta_0 + \gamma_0(\bw_i) + \varepsilon_i,
\end{equation*}
where $\btheta_0$ is the parameter of interest, $\gamma_0(\cdot)$ is an unknown function, and where $\E[\varepsilon_i|\bx_i,\bw_i]=0$. Defining $\dot{y}_{i,j}=y_i-y_j$ and $\dot{\bx}_{i,j}=\bx_i-\bx_j$, a pairwise difference estimator of $\btheta_0$ can be based on
\begin{equation*}
    m(\bz_i,\bz_j;\btheta) = m_{\mathtt{PLR}}(\bz_i,\bz_j;\btheta) = \frac{1}{2}(\dot{y}_{i,j}-\dot{\bx}_{i,j}'\btheta )^2.
\end{equation*}
Setting $\Theta=\R^k$, the minimization problem in \eqref{eq: Pairwise Difference Estimator} admits a closed form solution (provided that a non-negative kernel function is used), namely
\begin{equation*}
    \widehat{\btheta}_n = \left( \sum_{i<j} \dot{\bx}_{i,j}\dot{\bx}_{i,j}'K_{h_n}(\bw_i-\bw_j) \right)^{-1} \sum_{i<j} \dot{\bx}_{i,j}\dot{y}_{i,j}K_{h_n}(\bw_i-\bw_j).
\end{equation*}

\subsection{Partially Linear Logit Model}

The partially linear logit model studied here is of the form
\begin{equation*}
    y_i = \1 \{\bx_i'\btheta_0 + \gamma_0(\bw_i) + \varepsilon_i \geq 0 \},
\end{equation*}
where $\btheta_0$ is the parameter of interest, $\gamma_0(\cdot)$ is an unknown function, and where
\begin{equation*}
    \P\big[\varepsilon_i \leq u | \bx_i,\bw_i \big] = \Lambda(u), \qquad \Lambda(u) = \frac{\exp(u)}{1+\exp(u)}.
\end{equation*}
The parameter $\btheta_0$ can be estimated using a pairwise difference estimator with $\Theta=\R^k$ and
\begin{equation*}
	m(\bz_i,\bz_j;\btheta) = m_{\mathtt{PLL}}(\bz_i,\bz_j;\btheta) = -\1\{\dot{y}_{i,j}\neq 0\} \left(y_i\ln\Lambda(\dot{\bx}_{i,j}'\btheta) + y_j\ln\Lambda(-\dot{\bx}_{i,j}'\btheta)\right).
\end{equation*}
The minimization problem in \eqref{eq: Pairwise Difference Estimator} does not admit a closed form solution, but it is convex (provided that a non-negative kernel function is used) because $u\mapsto -\ln\Lambda(u)$ is.

\subsection{Partially Linear Tobit Model}

The partially linear censored regression model studied here is of the form
\begin{equation*}
    y_i =\max\{\bx_i'\btheta_0+\gamma_0(\bw_i)+\varepsilon_i,0\},
\end{equation*}
where $\btheta_0$ is the parameter of interest, $\gamma_0(\cdot)$ is an unknown function, $\bx_i\independent \varepsilon_i|\bw_i$, and the conditional distribution of $\varepsilon_i$ given $\bw_i$ admits a Lebesgue density. A pairwise difference estimator of $\btheta_0$ can be obtained by setting $\Theta=\R^k$ and employing

\begin{equation*}
    m(\bz_i,\bz_j;\btheta) = m_{\mathtt{PLT}}(\bz_i,\bz_j;\btheta) = \tilde{m}_{\mathtt{PLT}}(\bz_i,\bz_j;\btheta)-\tilde{m}_{\mathtt{PLT}}(\bz_i,\bz_j;\b0),
\end{equation*}
where
\begin{equation*}
    \tilde{m}_{\mathtt{PLT}}(\bz_i,\bz_j;\btheta) = 
    \begin{cases}
        |y_i| - \big( \dot{\bx}_{i,j}'\btheta + y_j \big) \sgn(y_i) & \text{if } \dot{\bx}_{i,j}'\btheta \leq -y_j \\
        \big|\dot{y}_{i,j} - \dot{\bx}_{i,j}'\btheta \big| & \text{if } -y_j < \dot{\bx}_{i,j}'\btheta < y_i \\
        |y_j| + \big(\dot{\bx}_{i,j}'\btheta-y_i \big) \sgn(y_j) & \text{if } y_i \leq \dot{\bx}_{i,j}'\btheta 
    \end{cases}.
\end{equation*}
Because $\tilde{m}_{\mathtt{PLT}}(\bz_i,\bz_j;\b0)$ does not depend on $\btheta$, its presence in $m_{\mathtt{PLT}}(\bz_i,\bz_j;\btheta)$ does not affect the minimization problem defining the estimator. Nevertheless, it is theoretically attractive to work with $m_{\mathtt{PLT}}$ rather than $\tilde{m}_{\mathtt{PLT}}$, as doing so allows for weaker regularity conditions for the existence of the expectation of the objective function.

For future reference, we note that $m_{\mathtt{PLT}}$ admits the alternative representation
\begin{equation*}
    m_{\mathtt{PLT}}(\bz_i,\bz_j;\btheta) =
    \begin{cases}
        \big|\dot{y}_{i,j}-\dot{\bx}_{i,j}'\btheta\big|-\big|\dot{y}_{i,j}\big| &\text{ if } y_i>0,y_j>0 \\
        \max\{ y_i-\dot{\bx}_{i,j}'\btheta, 0\}-y_i &\text{ if } y_i>0,y_j=0 \\
        \max\{ y_j+\dot{\bx}_{i,j}'\btheta, 0\}-y_j &\text{ if } y_i=0,y_j>0 \\
        0 &\text{ if } y_i=0,y_j=0
    \end{cases}.
\end{equation*}
The function $\btheta\mapsto m_{\mathtt{PLT}}(\bz_1,\bz_2;\btheta)$ is convex and therefore so is the minimization problem in \eqref{eq: Pairwise Difference Estimator} (provided that a non-negative kernel function is used).

\section{Distributional Approximation and Bootstrap Inference}\label{Section: Distributional Approximation and Bootstrap Inference}

As is standard in the literature, we generalize \eqref{eq: Pairwise Difference Estimator} slightly and define our estimator $\widehat{\btheta}_n=\widehat{\btheta}_n(h_n)$ to be any approximate minimizer of $\widehat{M}_n(\btheta;h_n)$, where
\begin{equation*}
     \widehat{M}_n(\btheta;h) = \binom{n}{2}^{-1}\sum_{i<j} m(\bz_i,\bz_j;\btheta) K_h(\bw_i-\bw_j).
\end{equation*}
To be specific, we require
\begin{equation*}
    \widehat{M}_n(\widehat{\btheta}_n(h);h) \leq \inf_{\btheta\in\Theta} \widehat{M}_n(\btheta;h) + o_{\P}( n^{-1} ).
\end{equation*}

The objective function $\widehat{M}_n$ is a sample counterpart of the function $M$ given by
\begin{equation*}
    M(\btheta;h) = \E\big[\widehat{M}_n(\btheta;h)\big]
                 = \E\big[ m(\bz_1,\bz_2;\btheta)  K_{h}(\bw_1-\bw_2)\big].
\end{equation*}
Under regularity conditions, this function approximates, as $h\downarrow 0$, a function $M_0$, which (does not depend on $K$ and) admits a unique minimizer, namely the parameter of interest $\btheta_0$.

For the purposes of analyzing $\widehat{\btheta}_n$ it is convenient to define $\btheta_n = \btheta(h_n)$, where
\begin{equation*}
    \btheta(h) \in \argmin_{\btheta\in\Theta} M(\btheta;h)
\end{equation*}
is interpretable as a (fixed-$h$) ``pseudo'' parameter. With the help of $\btheta_n$ we can decompose the estimation error $\widehat{\btheta}_n-\btheta_0$ into a (non-stochastic) ``bias'' component $\btheta_n-\btheta_0$ and a ``noise'' component $\widehat{\btheta}_n-\btheta_n$. Each component can be analyzed separately and in both cases the analysis will leverage convexity.

\subsection{Regularity Conditions}

The following assumption guarantees, among other things, that $\btheta_n$ is well defined for large $n$ and that the bias component $\btheta_n-\btheta_0$ vanishes asymptotically; for details, see Lemma \ref{Lemma: Existence and Convergence of theta(h)} of Section \ref{Section: A Useful Lemma}.

\begin{assumption}\label{Assumption: Convergence of M}
    \begin{enumerate}[(i)]
        \item \label{Assumption: Convergence of M - kernel function} The kernel function $K$ is a symmetric, bounded probability density.
        
        \item \label{Assumption: Convergence of M - convexity} $\Theta\subseteq\R^k$ is convex, $(\bz,\bar{\bz}) \mapsto m(\bz,\bar{\bz};\btheta)$ is permutation symmetric, and $\btheta \mapsto m(\bz_1,\bz_2;\btheta)$ is convex with probability one. 

        \item \label{Assumption: Convergence of M - density of w} The distribution of $\bw$ admits a Lebesgue density $f_{\bw}$, which is bounded and continuous on its support $\cW$.

        \item \label{Assumption: Convergence of M - well defined Mn} For each  $\btheta\in\Theta$,
        \begin{equation*}
            \E\left[\sup_{\bw_2\in\cW}|\E[m(\bz_1,\bz_2;\btheta)|\bw_1,\bw_2]|f_{\bw}(\bw_2)\right]<\infty
        \end{equation*}
        and (with probability one)
        \begin{equation*}
            \lim_{\bu\to \0} \E[m(\bz_1,\bz_2;\btheta)|\bw_1=\bw,\bw_2=\bw+\bu] = \E[m(\bz_1,\bz_2;\btheta)|\bw_1=\bw,\bw_2=\bw].
        \end{equation*}
        
        \item \label{Assumption: Convergence of M - well behaved M0} On $\Theta$, the function $M_0$ given by
        \begin{equation*}
            M_0(\btheta) = \int_{\cW} \E[m(\bz_1,\bz_2;\btheta)|\bw_1=\bw,\bw_2=\bw]f_{\bw}(\bw)^2d\bw
        \end{equation*}
        is uniquely minimized at an interior point $\btheta_0$.
        
    \end{enumerate}
\end{assumption}

The next assumption enables us to analyze the asymptotic properties of the noise component $\widehat{\btheta}_n-\btheta_n$. To accommodate examples (such as the partially linear Tobit model) where $\btheta \mapsto m(\bz_i,\bz_j;\btheta)$ is not fully differentiable, we assume the existence of derivative-like functions $\bs(\bz_i,\bz_j;\btheta)\in\R^k$ and $\bH(\bw_i,\bw_j;\btheta,\bt)\in\R^{k\times k}$ such that, for any direction $\bt \in \R^k$, the (remainder) terms
\begin{equation*}
    r_\bt(\btheta,\tau)=\frac{m(\bz_1,\bz_2;\btheta+\bt \tau)-m(\bz_1,\bz_2;\btheta)}{\tau} - \bs(\bz_1,\bz_2;\btheta)' \bt
\end{equation*}
and
\begin{equation*}
    R_\bt(\btheta,\tau)=\frac{\E[r_\bt(\btheta,\tau)|\bw_1,\bw_2]}{\tau} - \frac{1}{2}\bt'\bH(\bw_1,\bw_2;\btheta,\bt)\bt
\end{equation*}
are suitably small for $\btheta$ near $\btheta_0$, $\tau>0$ near zero, and $\bw_1 \approx \bw_2$. As further discussed below, functions $\bs$ and $\bH$ satisfying the following assumption exist (and are relatively easy to find) in each of our motivating examples.

\begin{assumption}\label{Assumption: Asymptotic Distribution}
    \begin{enumerate}[(i)]
        \item \label{Assumption: Asymptotic Distribution - differentiability} For each $\bt \in\R^k$, there is some $\delta>0$ such that
        \begin{align*}
            & \E\left[\sup_{ \tau\in (0,\delta),\|\btheta-\btheta_0\|<\delta, \bw_2 \in\cW}
                      \big|\E[ r_\bt(\btheta,\tau) | \bz_1,\bw_2 ] \big|f_{\bw}(\bw_2)^2 \right]< \infty,\\
            & \E\left[\sup_{ \tau\in (0,\delta),\|\btheta-\btheta_0\|<\delta, \bw_2 \in\cW}
                      \E[r_\bt(\btheta,\tau)^2 | \bw_1,\bw_2 ] f_{\bw}(\bw_2) \right] < \infty,\\
            &\E\left[\sup_{\tau\in (0,\delta),\|\btheta-\btheta_0\|<\delta,  \bw_2 \in\cW}
                     \left|R_\bt(\btheta,\tau)\right| f_{\bw}(\bw_2) \right]<\infty,
        \end{align*}
        and (with probability one)
        \begin{align*}
            & \lim_{\tau\downarrow0,(\btheta,\bu)\to(\btheta_0,\0)}
              \E[ r_\bt(\btheta,\tau) | \bz_1=\bz,\bw_2 =\bw+\bu ] = 0, \\
            & \lim_{\tau\downarrow0,(\btheta,\bu)\to(\btheta_0,\0)}
              \E[ r_\bt(\btheta,\tau)^2 | \bw_1=\bw,\bw_2=\bw+\bu ] = 0, \\
            & \lim_{\tau\downarrow0,(\btheta,\bu)\to(\btheta_0,\0)}
              \E[ R_\bt(\btheta,\tau) | \bw_1=\bw,\bw_2 =\bw+\bu ] = 0.
        \end{align*}

        \item \label{Assumption: Asymptotic Distribution - moment bounds} There is some $\delta>0$ and some function $b$ with
        \begin{equation*}
            \sup_{\|\btheta-\btheta_0\|<\delta} \|\bs(\bz_1,\bz_2;\btheta)\|\leq b(\bz_1)b(\bz_2),
        \end{equation*}
        such that
        \begin{equation*}
            \E[b(\bz)^4] + \sup_{\bw\in\cW}\E[b(\bz)^4|\bw]f_{\bw}(\bw)<\infty
        \end{equation*}
        and
        \begin{equation*}
            \E\left[\sup_{\|\btheta-\btheta_0\|<\delta, \bw_2\in\cW}                       \|\bH(\bw_1,\bw_2;\btheta,\bt)\| f_{\bw}(\bw_2) \right] < \infty \qquad \text{for each } \bt \in \R^k.
        \end{equation*}
        
        \item \label{Assumption: Asymptotic Distribution - variance ingredients} There exist functions $\bG_0,\bxi_0$, and $\bXi_0$ such that, for each $\bt\in\R^k$ (and with probability one),
        \begin{equation*}
            \lim_{(\btheta,\bu)\to(\btheta_0,\0)} \bH(\bw,\bw+\bu;\btheta,\bt) f_{\bw}(\bw) = \bG_0(\bw),
        \end{equation*}
        \begin{equation*}
            \lim_{(\btheta,\bu)\to(\btheta_0,\0)} -2\E[\bs(\bz_1,\bz_2;\btheta)|\bz_1=\bz,\bw_2=\bw+\bu ] f_{\bw}(\bw) = \bxi_0(\bz),
        \end{equation*}
        and
        \begin{equation*}
            \lim_{(\btheta,\bar{\btheta},\bu)\to (\btheta_0,\btheta_0,\0) } \E[\bs(\bz_1,\bz_2;\btheta)\bs(\bz_1,\bz_2;\bar{\btheta})'|\bw_1=\bw,\bw_2=\bw+\bu ] f_{\bw}(\bw)= \bXi_0(\bw).
        \end{equation*}

        \item \label{Assumption: Asymptotic Distribution - nonsingularity} $\bGamma_0 = \E[\bG_0(\bw)], \bSigma_0 = \E[\bxi_0(\bz)\bxi_0(\bz)']$, and $\E[\bXi_0(\bw)]$ are positive definite.
    \end{enumerate}
\end{assumption}

\subsection{Small Bandwidth Asymptotics}

Defining
\begin{equation*}
    \bV_n = \bV_n(h_n)=\bGamma_0^{-1}\left[ n^{-1}\bSigma_0 + \binom{n}{2}^{-1}h_n^{-d}\bDelta_0(K)\right]\bGamma_0^{-1}, \qquad \bDelta_0(K) = \E[ \bXi_0(\bw)] \int_{\R^d} K^2(\bu)d\bu,
\end{equation*}
and letting $\Phi_k$ denote the distribution function of a $k$-dimensional standard Gaussian random vector, we have the following result.

\begin{thm}\label{Theorem: Asymptotic Distribution}
    Suppose Assumptions \ref{Assumption: Convergence of M} and \ref{Assumption: Asymptotic Distribution} hold. If $n^2h_n^d\to\infty$ and if $h_n\to 0$, then
    \begin{equation*}
        \sup_{\bt\in\R^k} \left| \P\left[ \bV_n^{-1/2}(\widehat{\btheta}_n-\btheta_n ) \leq \bt \right] - \Phi_k(\bt) \right| \to 0.
    \end{equation*}
\end{thm}

Under the assumptions of Theorem \ref{Theorem: Asymptotic Distribution}, the convergence rate of $\widehat{\btheta}_n-\btheta_n$ is given by the magnitude of $\bV_n^{-1/2}$, namely
\begin{equation*}
    \rho_n = \sqrt{\min\left(n,\binom{n}{2}h_n^d\right)}.
\end{equation*}
Provided that the bias is ``small'' in the sense that $\rho_n\|\btheta_n-\btheta_0\|\to0$, Theorem \ref{Theorem: Asymptotic Distribution} therefore encompasses the following three distinct large-sample regimes:
\begin{itemize}
    \item \textit{Asymptotic Linearity}: If $nh_n^d\to\infty$, then $\sqrt{n}(\widehat{\btheta}_n-\btheta_0)$ satisfies \eqref{eq: Asymptotic Linearity} with $\bpsi_0 = \bGamma_0^{-1}\bxi_0$. In particular, $\sqrt{n}(\widehat{\btheta}_n-\btheta_0)$ converges in law to a centered Gaussian distribution with variance
    \begin{equation*}
        \lim_{n\to\infty} n \bV_n(h_n) = \bGamma_0^{-1} \bSigma_0 \bGamma_0^{-1}.
    \end{equation*}
    
    \item \textit{Root-$n$ Consistency without Asymptotic Linearity}: If $nh_n^d\to 2c \in(0,\infty)$, then $\widehat{\btheta}_n$ is not asymptotically linear, but it is $\sqrt{n}$-consistent, $\sqrt{n}(\widehat{\btheta}_n-\btheta_0)$ converging in law to a centered Gaussian distribution with variance
    \begin{equation*}
        \lim_{n\to\infty} n \bV_n(h_n) = \bGamma_0^{-1}\left[ \bSigma_0 + \frac{1}{c} \bDelta_0(K)\right]\bGamma_0^{-1}.
    \end{equation*}

    \item \textit{Slower than Root-$n$ Consistency}: If $nh_n^d\to 0$ (but $n^2h_n^d\to\infty$), then $\widehat{\btheta}_n$ is neither asymptotically linear nor $\sqrt{n}$-consistent, but $\sqrt{n^2h_n^d/2}(\widehat{\btheta}_n-\btheta_0)$ converges in law to a centered Gaussian distribution with variance
    \begin{equation*}
        \lim_{n\to\infty} \binom{n}{2}h_n^d \bV_n(h_n) = \bGamma_0^{-1}\bDelta_0(K)\bGamma_0^{-1}.
    \end{equation*}
\end{itemize}
The small bandwidth component (i.e., the term involving $\bDelta_0(K)$) in $\bV_n$ captures the additional uncertainty generated from increasing the localization of the observation pairs. Incorporating this component in the approximate variance is key to enabling us to replace the condition $nh_n^d\to\infty$ by the weaker condition $n^2h_n^d\to\infty$ when obtaining a Gaussian approximation. As demonstrated by \cite{Cattaneo-Farrell-Jansson-Masini_2025_JoE} in a related context, incorporating the small bandwidth component can furthermore lead to a more accurate distributional approximation (in a higher-order asymptotic sense) even under asymptotic linearity.

\subsection{Debiasing}

In Theorem \ref{Theorem: Asymptotic Distribution}, we centered the estimator $\widehat{\btheta}_n=\widehat{\btheta}_n(h_n)$ at $\btheta_n=\btheta(h_n)$ to circumvent bias issues. This section focuses on the bias term $\btheta_n-\btheta_0$ and introduces an automatic debiasing approach under the assumption that $\btheta_n-\btheta_0$ can be expanded in even powers of $h_n$. To be specific, we follow \citet[Section 3.3]{Honore-Powell_2005_Festschrift} and discuss debiasing under the following high-level condition.

\begin{assumption}\label{Assumption: Bias of thetahat}
    For some even $\smooth \geq0$, $\btheta(\cdot)$ admits $\bb_{2l}\in\R^k$ (for $l=1,\dots, \smooth/2$) such that
    \begin{equation*}
        \btheta(h) -\btheta_0 = \sum_{l=1}^{\smooth/2} \bb_{2l} h^{2l} + o(h^{\smooth})\qquad \text{as } h\downarrow 0.
    \end{equation*}
\end{assumption}

The ease with which Assumption \ref{Assumption: Bias of thetahat} can be verified depends on the magnitude of $\smooth$. For instance, Assumption \ref{Assumption: Convergence of M} implies that Assumption \ref{Assumption: Bias of thetahat} holds with $\smooth=0$. Under additional smoothness conditions and using symmetry of $K$, the following result gives conditions under which Assumption \ref{Assumption: Bias of thetahat} holds with $\smooth=2$. When stating the result, we employ the following standard multi-index notation: for $\bnu = (\nu_1,\ldots,\nu_d)'\in\Z_+^d$, $\bar\bw = (\bar w_1,\ldots,\bar w_d)'\in\R^d$, and a sufficiently smooth-in-$\bar{\bw}$ function $f(\bw,\bar{\bw}),$
\begin{equation*}
    \partial_{\bar{\bw}}^{\bnu} f(\bw,\bar{\bw}) = \frac{\partial^{|\bnu|}}{\partial \bar w_1^{\nu_1} \cdots \partial \bar w_d^{\nu_d}} f(\bw,\bar{\bw}), \qquad |\bnu| = \sum_{j=1}^d \nu_j.
\end{equation*}

\begin{prop}\label{Proposition: Bias Expansion L=2}
    Suppose Assumptions \ref{Assumption: Convergence of M}-\ref{Assumption: Asymptotic Distribution} hold and that
    \begin{enumerate}[(i)]
        \item $\int_{\R^d} \|\bu\|^2 K(\bu)d\bu<\infty$, and
        \item with probability one, $\bar{\bw} \mapsto \bvarphi(\bw,\bar{\bw}) = \E[\bs(\bz_1,\bz_2;\btheta_0)|\bw_1=\bw,\bw_2=\bar{\bw}]f_{\bw}(\bar{\bw})$ is twice continuously differentiable with $\E[\sup_{\bar{\bw}\in\cW}\|\partial_{\bar{\bw}}^{\bnu} \bvarphi(\bw,\bar{\bw})\|]<\infty$ for all $\bnu\in\Z_+^d$ with $|\bnu|\leq 2$.
    \end{enumerate}
    Then $\btheta(\cdot)$ admits a $\bb_2\in\R^k$ such that
    \begin{equation*}
        \btheta(h) - \btheta_0 =  \bb_2 h^2 + o(h^2) \qquad \text{ as } h\downarrow 0.
    \end{equation*}
\end{prop}

The proof of Proposition \ref{Proposition: Bias Expansion L=2} leverages convexity and may therefore be of independent interest. The convexity argument in question can furthermore be adapted to form the basis of a verification by induction of Assumption \ref{Assumption: Bias of thetahat} with $\smooth >2$. Details are provided in Section \ref{Section: Verifying Assumption 3}, which describes the induction step for general $\smooth$ and states explicit (smoothness) conditions under which Assumption \ref{Assumption: Bias of thetahat} holds with $\smooth =4$.

To describe the debiasing procedure based on generalized jackknifing, we maintain Assumption \ref{Assumption: Bias of thetahat}, define $c_0=1$, choose a non-negative integer $L$, and let $\bc=(c_0,\dots, c_{L})'$ be a vector of (distinct) positive constants such that the following vector is well defined:
\begin{equation*}
    \begin{pmatrix}
        \lambda_0(\bc) \\
        \lambda_1(\bc) \\
        \vdots \\
        \lambda_{L}(\bc)
    \end{pmatrix}
    = 
    \begin{pmatrix}
        1 & 1 &\dots &1 \\
        1 & c_1^2 &\dots &c_{L}^2 \\
        \vdots & &\ddots & \\
        1& c_1^{2L} & \dots & c_{L}^{2L}
    \end{pmatrix}^{-1}
    \begin{pmatrix}
        1 \\
        0 \\
        \vdots \\
        0
    \end{pmatrix}.
\end{equation*}
The debiased estimator is
\begin{equation*}
    \widetilde{\btheta}_n = \widetilde{\btheta}_n(\bc,h_n) = \sum_{l=0}^{L} \lambda_{l}(\bc) \widehat{\btheta}_n(h_{n,l}), \qquad h_{n,l}=c_lh_n,
\end{equation*}
the construction of which involves solving $L+1$ convex optimization problems. As defined, the debiased estimator is a generalization of the original pairwise difference estimator because if $L=0$, then $\bc=1=\lambda_0$ and therefore $\widetilde{\btheta}_n = \widehat{\btheta}_n$.

The next theorem generalizes Theorem \ref{Theorem: Asymptotic Distribution} by establishing the small bandwidth Gaussian approximation for $\widetilde{\btheta}_n$. To state the theorem, let
\begin{equation*}
    \bar{\btheta}_n = \bar{\btheta}_n(\bc,h_n) = \sum_{l=0}^{L} \lambda_l(\bc) \btheta(h_{n,l})
\end{equation*}
and
\begin{equation*}
    \bar{\bV}_n = \bar{\bV}_n(\bc,h_n) = \bGamma_0^{-1}\left[n^{-1}\bSigma_0 + \binom{n}{2}^{-1} h_n^{-d}\bDelta_0(\bar{K}) \right]\bGamma_0^{-1}, \qquad \bar{K}(\bu) = \bar{K}(\bu;\bc) = \sum_{l=0}^{L}\lambda_{l}(\bc)K_{c_l}(\bu).
\end{equation*}
As they should, the expressions have the feature that if $L=0$, then $\bar{\btheta}_n = \btheta_n$ and $\bar{\bV}_n=\bV_n$. Another noteworthy feature of the expressions is that debiasing via generalized jackknifing affects the variance $\bar{\bV}_n$ only through the kernel shape entering its small bandwidth component.

\begin{thm}\label{Theorem: Generalized Jackknifing}
    Suppose Assumptions \ref{Assumption: Convergence of M} and \ref{Assumption: Asymptotic Distribution} hold. If $n^2h_n^d\to\infty$ and if $h_n\to 0$, then
    \begin{equation*}
        \sup_{\bt\in\R^k} \left| \P\left[ \bar{\bV}_n^{-1/2}(\widetilde{\btheta}_n-\bar{\btheta}_n ) \leq \bt \right] - \Phi_k(\bt) \right| \to 0.
    \end{equation*}
    As a consequence, if also Assumption \ref{Assumption: Bias of thetahat} holds with $\smooth \geq 2L+2$, and if $nh_n^{4(L+1)}\to0$, then
    \begin{equation*}
        \sup_{\bt\in\R^k} \left| \P\left[ \bar{\bV}_n^{-1/2}(\widetilde{\btheta}_n-\btheta_0 ) \leq \bt \right] - \Phi_k(\bt) \right| \to 0
    \end{equation*}
\end{thm}

The magnitude of $\bar{\bV}_n^{-1/2}$ is the same as that of $\bV_n^{-1/2}$. With obvious modifications, the discussion of $\widehat{\btheta}_n$ following Theorem \ref{Theorem: Asymptotic Distribution} therefore applies to $\widetilde{\btheta}_n$, the only noteworthy difference being that (by design) the relevant ``small bias'' condition is different (and typically milder) in the case of $\widetilde{\btheta}_n$.

When $L=0$ (i.e.,\ one uses the original pairwise difference estimator), the second part of Theorem \ref{Theorem: Generalized Jackknifing} imposes $nh_n^4\to 0$, which coincides with the standard small bias condition in \eqref{eq: Bandwidth Conditions}, while the first part of the theorem still accommodates the small bandwidth distributional approximation.

It is worth noting that for $L\geq 1$, the equivalent kernel $\bar{K}$ is of higher order, even though the debiased estimator $\widetilde{\btheta}_n$ only employs estimators constructed using second-order kernels, thereby retaining the desired convexity for implementation. To be specific, if $\int_{\R^d} \|\bu\|^{2L+2} K(\bu)d\bu < \infty$, then $\bar{K}$ is of order $2L+2$ because
\begin{equation*}
    \int_{\R^d} \bar{K}(\bu)d\bu = \sum_{l=0}^{L}\lambda_l(\bc) \int_{\R^d} K_{c_{l}}(\bu)d\bu = \sum_{l=0}^{L}\lambda_l(\bc) = 1
\end{equation*}
and, for $\bnu=(\nu_1,\dots,\nu_d)'\in\Z^d_+$ with $0<|\bnu|\leq 2L+1$,
\begin{equation*}
    \int_{\R^d} \bu^{\bnu} \bar{K}(\bu)d\bu = \sum_{l=0}^{L} \lambda_l(\bc) \int_{\R^d} \bu^{\bnu} K_{c_l}(\bu)d\bu = \sum_{l=0}^{L} \lambda_l(\bc) c_l^{|\bnu|} \int_{\R^d} \bu^{\bnu} K(\bu)d\bu =0,
\end{equation*}
where the last equality uses the defining property of $\{\lambda_l(\bc)\}$ and symmetry of $K$, and where $\bu^{\bnu}$ denotes $\prod_{j=1}^d u_j^{\nu_j}$ for $\bu=(u_1,\dots,u_d)'\in\R^d$. 

\subsection{Bootstrapping}\label{Section: Bootstrapping}

To develop feasible inference procedures that do not require (explicit) estimation of $\bar{\bV}_n$, we consider nonparametric bootstrap-based approximations to the distribution of $\widetilde{\btheta}_n$. (Since $\widetilde{\btheta}_n=\widehat{\btheta}_n$ when $L=0$, results for $\widehat{\btheta}_n$ can be extracted by setting $L=0$ in what follows.)

Letting $\bz^*_{1,n},\dots,\bz^*_{n,n}$ denote a random sample from the empirical distribution of $\bz_1,\dots,\bz_n$, the defining property of $\widehat{\btheta}_n^*(h)$, the nonparametric bootstrap analogue of $\widehat{\btheta}_n(h)$, is the following:
\begin{equation*}
    \widehat{M}_n^*(\widehat{\btheta}_n^*(h);h) \leq \inf_{\btheta\in\Theta} \widehat{M}_n^*(\btheta;h)  + o_{\P}(n^{-1}),
\end{equation*}
where
\begin{equation*}
    \widehat{M}_n^*(\btheta;h) = \binom{n}{2}^{-1} \sum_{i<j}m(\bz^*_{i,n},\bz^*_{j,n};\btheta) K_h(\bw^*_{i,n}-\bw^*_{j,n}).
\end{equation*}
Similarly, the nonparametric bootstrap analogue of $\widetilde{\btheta}_n$ is
\begin{equation*}
    \widetilde{\btheta}_n^* = \widetilde{\btheta}_n^*(\bc)= \sum_{l=0}^{L} \lambda_l(\bc) \widehat{\btheta}_{n}^*(h_{n,l}).
\end{equation*}

The following theorem characterizes the large sample properties of $\widetilde{\btheta}_n^*-\widetilde{\btheta}_n$, the bootstrap counterpart of $\widetilde{\btheta}_n-\bar{\btheta}_n$. In perfect analogy with the results in \cite{Cattaneo-Crump-Jansson_2014b_ET}, we find that the bootstrap distribution estimator is consistent only when $nh_n^d\to\infty$, but otherwise exhibits a variance inflation making the distributional approximation inconsistent. To state the result, let $\P_n^*[\cdot]$ denote $\P[\cdot | \bz_1,\ldots,\bz_n]$, let $\to_\P$ denote convergence in probability, and define
\begin{equation*}
    \bar{\bV}_n^* = \bar{\bV}_n^*(\bc,h_n) = \bGamma_0^{-1}\left[n^{-1}\bSigma_0 + 3\binom{n}{2}^{-1} h_n^{-d}\bDelta_0(\bar{K}) \right]\bGamma_0^{-1}.
\end{equation*}

\begin{thm}\label{Theorem: Bootstrapping}
    Suppose Assumptions \ref{Assumption: Convergence of M}-\ref{Assumption: Asymptotic Distribution} hold and that, for $\btheta$ near $\btheta_0$ (and with probability one), $m(\bz,\bz;\btheta)=0$ and $\bs(\bz,\bz;\btheta)=\0$. If $n^2h_n^d\to\infty$ and if $h_n\to0$, then
    \begin{equation*}
        \sup_{\bt\in\R^k} \left| \P_n^*\left[ \bar{\bV}_n^{*-1/2}(\widetilde{\btheta}_n^*-\widetilde{\btheta}_n ) \leq \bt \right] - \Phi_k(\bt) \right| \to_\P 0.
    \end{equation*}
\end{thm}

Because $\bar{\bV}_n^{-1}\bar{\bV}_n^*\to \bI_k$ if and only if $nh_n^d\to\infty$ (where $\bI_k$ denotes the $k$-dimensional identity matrix), under the assumptions of Theorem \ref{Theorem: Bootstrapping},
\begin{equation*}
    \sup_{\bt\in\R^k} \left| \P_n^*\left[ \widetilde{\btheta}_n^*-\widetilde{\btheta}_n \leq \bt \right] - \P\left[ \widetilde{\btheta}_n - \bar{\btheta}_n \leq \bt \right] \right| \to_\P 0
\end{equation*}
if and only if $nh_n^d\to\infty$. In particular, if $\liminf_{n\to\infty} nh_n^d<\infty$, then the nonparametric bootstrap is inconsistent, albeit conservative in the sense that the (approximate) variance under the bootstrap distribution is larger than the (approximate) variance of the asymptotic distribution: $\bar{\bV}_n^* > \bar{\bV}_n$ in a positive definite sense.

The variance inflation problem associated with the nonparametric bootstrap under the small bandwidth regime can be easily fixed by appropriately rescaling the bandwidth used for the bootstrap implementation of the pairwise difference estimator. Employing
\begin{equation*}
    \breve{\btheta}_n^* = \breve{\btheta}_n^*(\bc,h_n)= \sum_{l=0}^{L} \lambda_l(\bc) \widehat{\btheta}_{n}^*(3^{1/d}h_{n,l})
\end{equation*}
and centering its distribution at
\begin{equation*}
    \breve{\btheta}_n = \breve{\btheta}_n(\bc,h_n)= \sum_{l=0}^{L} \lambda_l(\bc) \widehat{\btheta}_{n}(3^{1/d}h_{n,l})
\end{equation*}
automatically adjusts the bootstrap variance, leading to a consistent distributional approximation. Indeed, the following result is an immediate consequence of Theorems \ref{Theorem: Generalized Jackknifing} and \ref{Theorem: Bootstrapping}.

\begin{coro}\label{Corollary: Bootstrapping}
   If the assumptions of Theorem \ref{Theorem: Bootstrapping} hold, then
   \begin{equation*}
        \sup_{\bt\in\R^k} \left| \P_n^*\left[ \breve{\btheta}_n^* - \breve{\btheta}_n \leq \bt \right] - \P\left[ \widetilde{\btheta}_n - \bar{\btheta}_n \leq \bt \right] \right| \to_\P 0.
    \end{equation*}
    As a consequence, if also Assumption \ref{Assumption: Bias of thetahat} holds with $\smooth \geq 2L+2$, and if $nh_n^{4(L+1)}\to0$, then
    \begin{equation*}
        \sup_{\bt\in\R^k} \left| \P_n^*\left[ \breve{\btheta}_n^* - \breve{\btheta}_n \leq \bt \right] - \P\left[ \widetilde{\btheta}_n - \btheta_{0} \leq \bt \right] \right| \to_\P 0.
    \end{equation*}    
\end{coro}

The statement of Corollary \ref{Corollary: Bootstrapping} emphasizes the rate-adaptive nature of the consistency property enjoyed by the bootstrap distributional approximation.

\subsection{Discussion}\label{Section: Discussion}

Our results provide a simple practical message: bootstrap-based inference for convex pairwise difference estimators is highly sensitive to bandwidth choice under conventional implementations, but this sensitivity can be substantially reduced by combining debiasing with a bandwidth rescaling. In particular, the classical bootstrap is valid only over a relatively narrow range of bandwidth sequences, whereas generalized jackknifing enlarges the set of admissible large bandwidths, small bandwidth asymptotics justifies an implementation enlarging the set of admissible small bandwidths, and the combination of the two yields the broadest range of valid inference.

To summarize these conclusions, suppose $\smooth\geq 2L+2$, let $\alpha \in (0,1)$ and $\ba\in\R^k$ be fixed, and consider the following family of percentile bootstrap confidence intervals \citep[in the terminology of][]{vanderVaart_1998_Book} for $\ba'\btheta_0$:
\begin{align*}
    \mathsf{CI}_{n,1-\alpha}^*(B,L)
    = \Big[\ba'\widetilde{\btheta}_n(\bc,h_n) - \widetilde{q}^*_{1-\alpha/2,n}(Bh_n,L) ~ , ~
           \ba'\widetilde{\btheta}_n(\bc,h_n) - \widetilde{q}^*_{\alpha/2,n}(Bh_n,L) \Big],
\end{align*}
where
\begin{align*}
    \widetilde{q}^*_{t,n}(h,L)
    = \inf\big\{ q \in \R : \P_n^* [\ba'\widetilde{\btheta}_n^*(\bc,h) - \ba'\widetilde{\btheta}_n(\bc,h)\leq q] \geq t \big\}.
\end{align*}
The notation highlights the two tuning parameters that vary across the different inference procedures: given a choice of bandwidth $h$ used in the point estimator, the scaling constant $B$ used in the bootstrap approximation and the debiasing order $L=\dim(\bc)-1$, which together with $\bc$ determines the generalized jackknife weights.

To sharpen the practical interpretation of these conclusions, it is useful to view each procedure through its coverage function
\begin{align*}
    \mathsf{CF}_{n,1-\alpha}(B,L)=\P\!\left[\ba'\btheta_0\in \mathsf{CI}_{n,1-\alpha}^*(B,L)\right].
\end{align*}
Our theory implies that the main differences across procedures can be understood in terms of how $\mathsf{CF}_{n,1-\alpha}(B,L)$ behaves as the bandwidth sequence moves from the small-bandwidth region to the large-bandwidth region. In particular, the classical bootstrap tends to become conservative when the bandwidth is too small, because the bootstrap variance is no longer correctly matched under small bandwidth asymptotics, and it tends to undercover when the bandwidth is too large, because smoothing bias becomes non-negligible. Debiasing mitigates the latter problem by shifting the onset of bias-driven undercoverage toward larger bandwidths, whereas bootstrap rescaling mitigates the former problem by restoring correct variance matching in the small-bandwidth region. Combining the two therefore produces the flattest coverage profile, in the sense that the coverage function remains close to its nominal level over the widest range of bandwidth sequences.

We formally compare the four implementations, ordered from the conventional procedure to our recommended method.
\begin{itemize}
    \item \textit{Classical Method}. This corresponds to $(B,L)=(1,0)$. If Assumptions \ref{Assumption: Convergence of M}-\ref{Assumption: Bias of thetahat} hold and
    \begin{align*}%\label{eq: BW conditions -- Classical}
        n h_n^d\to\infty
        \qquad\text{and}\qquad
        n h_n^{4} \to 0,
    \end{align*}
    then
    \begin{align*}
        \lim_{n \to \infty} \P\left[\ba'\btheta_0 \in \mathsf{CI}_{n,1-\alpha}^*(1,0) \right] = 1 - \alpha.
    \end{align*}
    The admissible range of bandwidth sequences is narrow, and requires $d<4$. In terms of the coverage function, the classical percentile bootstrap has two distinct failure modes. When the bandwidth is too small, so that $n h_n^d$ is no longer large while $n^2 h_n^d\to\infty$ still holds, the procedure becomes conservative: $\mathsf{CF}_{n,1-\alpha}(1,0)$ rises above its nominal level because the bootstrap variance is too large relative to the sampling variance. When the bandwidth is too large, so that the bias condition fails, $\mathsf{CF}_{n,1-\alpha}(1,0)$ falls below its nominal level because smoothing bias is no longer negligible. Thus, the classical method delivers accurate coverage only over a narrow range of bandwidth choices.

    \item \textit{Classical Debiased Method}. This corresponds to $(B,L)=(1,L)$ for some $L\geq 1$. If Assumptions \ref{Assumption: Convergence of M}-\ref{Assumption: Bias of thetahat} hold and
    \begin{align*}%\label{eq: BW conditions -- Classical Debiased}
        n h_n^d\to\infty
        \qquad\text{and}\qquad
        n h_n^{4(L+1)} \to 0,
    \end{align*}
    then
    \begin{align*}
        \lim_{n \to \infty} \P\left[\ba'\btheta_0 \in \mathsf{CI}_{n,1-\alpha}^*(1,L) \right] = 1 - \alpha.
    \end{align*}
    The range of allowable bandwidth sequences is wider on the large-bandwidth side than for the classical method, and requires $d<4(L+1)$. Relative to the classical method, debiasing leaves the small-bandwidth behavior essentially unchanged: if the bandwidth is too small, the procedure may still be conservative because the bootstrap variance mismatch remains. Its gain is instead on the large-bandwidth side, where the onset of bias-driven undercoverage is pushed outward from the $n h_n^4$ scale to the weaker $n h_n^{4(L+1)}$ scale. Consequently, $\mathsf{CF}_{n,1-\alpha}(1,L)$ stays close to nominal coverage over a wider range of large bandwidth sequences than $\mathsf{CF}_{n,1-\alpha}(1,0)$.

    \item \textit{Small Bandwidth Method}. This corresponds to $(B,L)=(3^{1/d},0)$. If Assumptions \ref{Assumption: Convergence of M}-\ref{Assumption: Bias of thetahat} hold and
    \begin{align*}%\label{eq: BW conditions -- Small Bandwidth}
        n^2 h_n^d\to\infty
        \qquad\text{and}\qquad
        n h_n^{4} \to0,
    \end{align*}
    then
    \begin{align*}
        \lim_{n \to \infty} \P\left[\ba'\btheta_0 \in \mathsf{CI}_{n,1-\alpha}^*(3^{1/d},0) \right] = 1 - \alpha.
    \end{align*}
    The range of allowable bandwidth sequences is wider on the small-bandwidth side than for the classical approach, and requires $d<8$. Relative to the classical method, bootstrap rescaling corrects the variance mismatch responsible for small-bandwidth conservativeness. As a result, $\mathsf{CF}_{n,1-\alpha}(3^{1/d},0)$ remains close to its nominal level throughout the admissible small-bandwidth region. Its limitation is on the large-bandwidth side: because the procedure is not debiased, coverage still deteriorates once the bias condition $n h_n^4\to 0$ fails.

    \item \textit{Small Bandwidth Debiased Method}. This corresponds to $(B,L)=(3^{1/d},L)$ for some $L\geq 1$. If Assumptions \ref{Assumption: Convergence of M}-\ref{Assumption: Bias of thetahat} hold and
    \begin{align*}%\label{eq: BW conditions -- Small Bandwidth Debiased}
        n^2 h_n^d\to\infty
        \qquad\text{and}\qquad
        n h_n^{4(L+1)} \to 0,
    \end{align*}
    then
    \begin{align*}
        \lim_{n \to \infty} \P\left[\ba'\btheta_0 \in \mathsf{CI}_{n,1-\alpha}^*(3^{1/d},L) \right] = 1 - \alpha.
    \end{align*}
    The range of allowable bandwidth sequences is the widest among the four procedures, and requires $d<8(L+1)$. This method combines the two improvements described above. On the small-bandwidth side, bootstrap rescaling removes the variance mismatch that would otherwise induce conservativeness. On the large-bandwidth side, debiasing delays the onset of bias-driven undercoverage. Accordingly, $\mathsf{CF}_{n,1-\alpha}(3^{1/d},L)$ remains close to its nominal level over the broadest range of bandwidth sequences considered in this paper.
\end{itemize}

Figure \ref{fig:coverage-schematic} presents a schematic plot summarizing the main theoretical conclusions. Overall, the comparison suggests a clear practical recommendation: bootstrap inference based on small bandwidth and debiasing corrections provides the most robust coverage accuracy as a function of bandwidth choice. Thus, for empirical settings where bandwidth choice is uncertain or data-driven, this method provides the most reliable default.

\begin{figure}[!htbp]
\centering
\begin{tikzpicture}
\begin{axis}[
    width=13.8cm,
    height=8cm,
    xmin=0, xmax=1.02,
    ymin=0.02, ymax=1.05,
    axis lines=left,
    xlabel={$h_n$},
    xlabel style={at={(axis description cs:0.985,-0.025)}, anchor=north},
    ylabel={Coverage},
    ylabel style={at={(axis description cs:-0.035,0.5)}, anchor=center},
    xtick=\empty,
    ytick={0.8,1.0},
    yticklabels={$1-\alpha$,1},
    clip=false,
    samples=500
]

% Threshold locations
\def\xsmall{0.05}
\def\xclassical{0.4}
\def\xbiasa{0.55}   % moved left: where blue/red begin to diverge
\def\xbiasb{0.70}   % moved left: where cyan/orange begin to diverge

% Nominal level across full axis
\addplot[dashed, gray!70, line width=0.7pt, domain=0:1.02] {0.8};

% Reference lines
\addplot[dashed, gray!70, line width=0.7pt] coordinates {(\xsmall,0.02) (\xsmall,1.02)};
\addplot[dashed, gray!70, line width=0.7pt] coordinates {(\xclassical,0.02) (\xclassical,1.02)};
\addplot[dashed, gray!70, line width=0.7pt] coordinates {(\xbiasa,0.02) (\xbiasa,1.02)};
\addplot[dashed, gray!70, line width=0.7pt] coordinates {(\xbiasb,0.02) (\xbiasb,1.02)};

% Yellow Dashed-Dotted: classical
\addplot[
    color=orange!95!black,
    dashdotted,
    line width=1pt,
    domain=\xsmall:0.8
]
{0.792
 + 0.175*(1 - 1/(1+exp(-(x-0.31)/0.055)))
 - 0.845*(1/(1+exp(-(x-0.72)/0.050)))
};

% Cyan Dashed: classical debiased
\addplot[
    color=cyan!70!black,
    dashed,
    line width=1pt,
    domain=\xsmall:0.95
]
{0.807
 + 0.160*(1 - 1/(1+exp(-(x-0.295)/0.055)))
 - 0.820*(1/(1+exp(-(x-0.865)/0.048)))
};

% Red Dotted: small-bandwidth, slightly below 1-alpha
\addplot[
    color=red!85!black,
    densely dotted,
    line width=1pt,
    domain=\xsmall:.8
]
{0.797
 + 0.008*exp(-((x-0.65)^2)/(0.15^2))
 - 0.845*(1/(1+exp(-(x-0.715)/0.050)))
};

% Blue Solid: small-bandwidth debiased, slightly above 1-alpha
\addplot[
    color=blue!85!black,
    solid,
    line width=1pt,
    domain=\xsmall:.96
]
{0.805
 + 0.004*exp(-((x-0.55)^2)/(0.19^2))
 - 0.820*(1/(1+exp(-(x-0.862)/0.048)))
};

% Labels under each vertical line
\node[anchor=north, font=\scriptsize] at (axis cs:\xsmall,-0.01) {$n^2 h_n^d \to \infty$};
\node[anchor=north, font=\scriptsize] at (axis cs:\xclassical,-0.01) {$n h_n^d \to \infty$};
\node[anchor=north, font=\scriptsize] at (axis cs:\xbiasa,-0.01) {$n h_n^{4} \to 0$};
\node[anchor=north, font=\scriptsize] at (axis cs:\xbiasb,-0.01) {$n h_n^{8} \to 0$};

% Manual legend moved left to almost touch xsmall
\draw[orange!95!black, dashdotted, line width=1pt]
    (axis cs:0.075,0.48) -- (axis cs:0.135,0.48);
\node[anchor=west, font=\tiny] at (axis cs:0.14,0.48)
    {$\mathsf{CF}_{n,1-\alpha}(1,0)$};

\draw[cyan!70!black, dashed, line width=1pt]
    (axis cs:0.075,0.42) -- (axis cs:0.135,0.42);
\node[anchor=west, font=\tiny] at (axis cs:0.14,0.42)
    {$\mathsf{CF}_{n,1-\alpha}(1,1)$};

\draw[red!85!black, densely dotted, line width=1pt]
    (axis cs:0.075,0.36) -- (axis cs:0.135,0.36);
\node[anchor=west, font=\tiny] at (axis cs:0.14,0.36)
    {$\mathsf{CF}_{n,1-\alpha}(3^{1/d},0)$};

\draw[blue!85!black, solid, line width=1pt]
    (axis cs:0.075,0.30) -- (axis cs:0.135,0.30);
\node[anchor=west, font=\tiny] at (axis cs:0.14,0.30)
    {$\mathsf{CF}_{n,1-\alpha}(3^{1/d},1)$};

\end{axis}
\end{tikzpicture}
\caption{Schematic coverage functions for four bootstrap confidence interval procedures.}
\label{fig:coverage-schematic}
\end{figure}

\section{Proofs and Other Technical Results}\label{Section: Proofs and Other Technical Results}

\subsection{A Useful Lemma}\label{Section: A Useful Lemma}

The following lemma is used in the proofs of Theorems \ref{Theorem: Asymptotic Distribution}-\ref{Theorem: Bootstrapping}.

\begin{lem}\label{Lemma: Existence and Convergence of theta(h)}
    Suppose that Assumption \ref{Assumption: Convergence of M} holds. Then $\argmin_{\btheta\in\Theta}M(\btheta;h)$ is non-empty for $h>0$ near zero and
    \begin{equation*}
        \btheta(h) - \btheta_0 = o(1) \qquad \text{as } h \downarrow 0.
    \end{equation*}
    If also Assumption \ref{Assumption: Asymptotic Distribution} holds, then $\E[\bs(\bz_1,\bz_2;\btheta(h) ) K_h(\bw_1-\bw_2)]=\0$ for $h>0$ near zero.
\end{lem}

\begin{proof}
    For every $\btheta \in \Theta$,
    \begin{align*}
        M(\btheta;h) &= \int_\cW \int_{\R^d} \E[m(\bz_1,\bz_2;\btheta)|\bw_1=\bw,\bw_2=\bw-\bu h] f_\bw(\bw) f_\bw(\bw - \bu h) K(\bu) d\bu d\bw \\
        &\to \int_\cW \E[m(\bz_1,\bz_2;\btheta)|\bw_1=\bw,\bw_2=\bw] f_\bw(\bw)^2 d\bw = M_0(\btheta) \qquad \text{as } h \downarrow 0,
    \end{align*}
    the convergence being uniform on compact subsets of $\Theta$ because $\btheta \mapsto M(\btheta;h)$ is convex \citep[e.g.,][Lemma 1]{Hjort-Pollard_1993}.
    
    Take any $\epsilon > 0$ with $\Theta^\epsilon_0 = \{\btheta \in \R^k: \|\btheta-\btheta_0\| \leq \epsilon\} \subseteq \Theta$. By the preceding paragraph,
    \begin{equation*}
        \sup_{\btheta \in \Theta^\epsilon_0} |M(\btheta;h) - M_0(\btheta)| \leq \frac{1}{2} \left(\inf_{\btheta:\|\btheta-\btheta_0\|=\epsilon} M_0(\btheta) - M_0(\btheta_0)\right)
    \end{equation*}
    for $h>0$ near zero. For any such $h$ and any $\btheta \in \Theta \setminus \Theta^\epsilon_0$, we have
    \begin{equation*}
        \epsilon_{\btheta} = \frac{\epsilon}{\|\btheta-\btheta_0\|} \in (0,1),
    \end{equation*}
    and therefore, by convexity of $\btheta \mapsto M(\btheta;h)$,
    \begin{equation*}
        M(\epsilon_{\btheta} \btheta + (1-\epsilon_{\btheta}) \btheta_0;h) \leq \epsilon_{\btheta} M(\btheta;h) + (1-\epsilon_{\btheta}) M(\btheta_0;h),
    \end{equation*}
    which rearranges as
    \begin{equation*}
        M(\btheta;h) - M(\btheta_0;h)
        \geq \frac{1}{\epsilon_{\btheta}} [M(\epsilon_{\btheta} \btheta + (1-\epsilon_{\btheta}) \btheta_0;h) - M(\btheta_0;h)] \geq 0. 
    \end{equation*}
    As a consequence,
    \begin{equation*}
        \inf_{\btheta \in \Theta}M(\btheta;h) = \inf_{\btheta \in \Theta^\epsilon_0}M(\btheta;h) = \min_{\btheta \in \Theta^\epsilon_0}M(\btheta;h),
    \end{equation*}
    where the last equality uses continuity of $\btheta \mapsto M(\btheta;h)$ and compactness of $\Theta^\epsilon_0$.
    
    The above argument shows in particular that $\btheta(h) \in \Theta^\epsilon_0$ for $h>0$ near zero.

    If also Assumption \ref{Assumption: Asymptotic Distribution} holds, then, for $\btheta$ near $\btheta_0$, $h>0$ near zero, and for any $\bt \in \R^k$,
    \begin{align*}
        \left| \frac{M(\btheta+\bt\tau;h)-M(\btheta;h)}{\tau} - \E[\bs(\bz_1,\bz_2;\btheta) K_h(\bw_1-\bw_2)]' \bt \right| &\leq \left| \E[r_\bt(\btheta,\tau) K_h(\bw_1-\bw_2)] \right| \\
        &\to 0 \qquad \text{as } \tau \downarrow 0,
    \end{align*}
    implying that for $h>0$ near zero, $\btheta \mapsto M(\btheta;h)$ is (directionally) differentiable near $\btheta_0$, the directional derivative $\E[\bs(\bz_1,\bz_2;\btheta) K_h(\bw_1-\bw_2)]' \bt$ being zero when $\btheta = \btheta(h)$ because $\btheta(h)$ minimizes $M(\btheta;h).$ 
\end{proof}

\subsection{Proof of Theorems \ref{Theorem: Asymptotic Distribution} and \ref{Theorem: Generalized Jackknifing}}

Theorem \ref{Theorem: Asymptotic Distribution} can be obtained from Theorem \ref{Theorem: Generalized Jackknifing} by setting $L=0$ in the latter, so it suffices to prove Theorem \ref{Theorem: Generalized Jackknifing}. To do so, for $l \in\{0,\dots,L\}$, let $\widehat{\btheta}_{n,l}=\widehat{\btheta}_n(h_{n,l}),\btheta_{n,l}=\btheta(h_{n,l}),$ and
\begin{equation*}
    \widehat{\bU}_{n,l} = \binom{n}{2}^{-1} \sum_{i<j}  \bs^{\mu}_{n,l}(\bz_i,\bz_j), \qquad \bs^{\mu}_{n,l}(\bz_i,\bz_j)= \bs_{n,l}(\bz_i,\bz_j) - \E[\bs_{n,l}(\bz_1,\bz_2)],
\end{equation*}
where
\begin{equation*}
    \bs_{n,l}(\bz_i,\bz_j)=\bs(\bz_i,\bz_j;\btheta_{n,l})K_{h_{n,l}}(\bw_i-\bw_j).
\end{equation*}
By Lemma \ref{Lemma: Existence and Convergence of theta(h)}, $\lim_{n \to \infty}\btheta_{n,l}=\btheta_0$ and $\E[\bs_{n,l}(\bz_i,\bz_j)]=0$ for large $n$.

Suppose that
\begin{equation}\label{eq: Hjort-Pollard stochastic expansion}
    \widehat{\btheta}_{n,l}-\btheta_{n,l} = -\bGamma_0^{-1}\widehat{\bU}_{n,l} + o_{\P}(\rho_n^{-1}) \quad \text{for} \quad l \in\{0,\dots,L\}.
\end{equation}
Then
\begin{equation*}
    \widetilde{\btheta}_n-\bar{\btheta}_n = -\bGamma_0^{-1}\widetilde{\bU}_n + o_{\P}(\rho_n^{-1}),
\end{equation*}
where
\begin{equation*}
    \widetilde{\bU}_n = \sum_{l=0}^{L} \lambda_l(\bc)\widehat{\bU}_{n,l}
\end{equation*}
satisfies
\begin{equation}\label{eq: U-statistic asymptotic normality}
    \left[ n^{-1}\bSigma_0 + \binom{n}{2}^{-1}h_n^{-d}\bDelta_0(\bar{K})\right]^{-1/2} \widetilde{\bU}_n \rightsquigarrow \cN (\0_{k \times 1},\bI_k)
\end{equation}
because, letting $\sum_i$ denote $\sum_{i=1}^n$,
\begin{equation*}
    \widetilde{\bU}_n = \widetilde{\bL}_n + \widetilde{\bW}_n,
\end{equation*}
where
\begin{equation*}
    \widetilde{\bL}_n = \frac{2}{n} \sum_i \widetilde{\bell}_n(\bz_i), \qquad \widetilde{\bell}_n(\bz_i)=\sum_{l=0}^{L} \lambda_l(\bc) \E[\bs^{\mu}_{n,l}(\bz_i,\bz_j)|\bz_i] \qquad (j \neq i),
\end{equation*}
and
\begin{equation*}
    \widetilde{\bW}_n = \binom{n}{2}^{-1} \sum_{i<j} \widetilde{\bomega}_n(\bz_i,\bz_j), \qquad \widetilde{\bomega}_n(\bz_i,\bz_j)=\sum_{l=0}^{L} \lambda_l(\bc) \big[\bs^{\mu}_{n,l}(\bz_i,\bz_j) - \widetilde{\bell}_n(\bz_i) - \widetilde{\bell}_n(\bz_j)\big],
\end{equation*}
satisfy
\begin{equation*}
    \left(\begin{array}{c} \sqrt{n}\widetilde{\bL}_n \\  \sqrt{\binom{n}{2}h_n^d} \widetilde{\bW}_n \end{array}\right) \rightsquigarrow \cN\left( \left[\begin{array}{c} \0_{k \times 1} \\ \0_{k \times 1} \end{array}\right], \left[\begin{array}{cc} \bSigma_0 &\0_{k \times k} \\ \0_{k \times k} &\bDelta_0(\bar{K})  \end{array}\right] \right),
\end{equation*}
as can be shown by means of the Cram{\'e}r-Wold device and the central limit theorem of \cite{Heyde-Brown_1970}, the latter being applicable because it follows from routine calculations that for every $\bmu_1,\bmu_2 \in \R^k$, we have
\begin{equation*}
    \sigma_n^2 = \sum_i\V[g_{i,n}] = \bmu_1'\bSigma_0\bmu_1 + \bmu_2'\bDelta_0(\bar{K})\bmu_2 + o(1), 
\end{equation*}
\begin{equation*}
    \sum_i \E[g_{i,n}^4] = o(1),
\end{equation*}
and
\begin{equation*}
    \V\left[\sum_i\sigma_{i,n}^2-\sigma_n^2\right] = o(1),\qquad \sigma_{i,n}^2=\V[g_{i,n}|\bz_1,\dots,\bz_{i-1}],
\end{equation*}
where
\begin{equation*}
    g_{i,n} = g_{i,n}(\bmu) = \frac{2}{\sqrt{n}} \bmu_1'\widetilde{\bell}_n(\bz_i) + \sqrt{\binom{n}{2}^{-1}h_n^{d}}\sum_{j=1}^{i-1}\bmu_2'\widetilde{\bomega}_{n}(\bz_i,\bz_j).
\end{equation*}
The proof of Theorem \ref{Theorem: Generalized Jackknifing} can therefore be completed by verifying \eqref{eq: Hjort-Pollard stochastic expansion}.

To do so, we leverage convexity. For any $l \in \{0,\dots,L\}$ and any $\bt \in \R^k$, it can be shown that
\begin{equation*}
    \lim_{\tau \downarrow 0,h \downarrow 0,\btheta \to \btheta_0} \E \big[ \E[ r_{\bt}(\btheta,\tau) K_h(\bw_1-\bw_2) | \bz_1 ] ^2 \big] = 0
\end{equation*}
and
\begin{equation*}
    \lim_{\tau \downarrow 0,h \downarrow 0,\btheta \to \btheta_0} h^d \E[ r_{\bt}(\btheta,\tau)^2 K_h(\bw_1-\bw_2)^2 ] = 0,
\end{equation*}
and it therefore follows from a Hoeffding decomposition that
\begin{align*}
	&\rho_n^2\Big[\widehat{M}_n(\btheta_{n,l}+\bt \rho_n^{-1};h_{n,l}) - \widehat{M}_n(\btheta_{n,l};h_{n,l}) \Big] \\
    &=  \rho_n^2 [M(\btheta_{n,l}+\bt \rho_n^{-1};h_{n,l}) - M(\btheta_{n,l};h_{n,l})] + \bt'\rho_n\widehat{\bU}_{n,l} + o_\P(1).
\end{align*}
Moreover,
\begin{equation*}
    \rho_n^2 [M(\btheta_{n,l}+\bt \rho_n^{-1};h_{n,l}) - M(\btheta_{n,l};h_{n,l}) ] \to \frac{1}{2} \bt' \bGamma_0 \bt,
\end{equation*}
and proceeding as in the proof of \eqref{eq: U-statistic asymptotic normality} it can be shown that $\rho_n\widehat{\bU}_{n,l} = O_\P(1)$. Because $\bGamma_0$ is positive definite and because $\bt \mapsto \widehat{M}_n(\btheta_{n,l}+\bt \rho_n^{-1};h_{n,l})$ is convex (almost surely), the corollary following \citet[Lemma 2]{Hjort-Pollard_1993} implies that \eqref{eq: Hjort-Pollard stochastic expansion} holds.

\subsection{Proof of Theorem \ref{Theorem: Bootstrapping}}

The proof of Theorem \ref{Theorem: Bootstrapping} is a natural bootstrap analog of the proof of Theorem \ref{Theorem: Generalized Jackknifing}.

For $l \in\{0,\dots,L\}$, let $\widehat{\btheta}^*_{n,l}=\widehat{\btheta}^*_n(h_{n,l})$ and
\begin{equation*}
    \widehat{\bU}^*_{n,l} = \binom{n}{2}^{-1} \sum_{i<j}  \bs^{\mu,*}_{n,l}(\bz^*_{i,n},\bz^*_{j,n}), \qquad \bs^{\mu,*}_{n,l}(\bz^*_{i,n},\bz^*_{j,n})= \bs_{n,l}(\bz^*_{i,n},\bz^*_{j,n}) - \E^*_n[\bs_{n,l}(\bz^*_{1,n},\bz^*_{2,n})],
\end{equation*}
where $\E^*_n[\cdot]$ denotes $\E[\cdot|\bz_1,\dots,\bz_n]$.

It suffices to show that
\begin{equation}\label{eq: Hjort-Pollard stochastic expansion - bootstrap}
    \widehat{\btheta}^*_{n,l}-\btheta_{n,l} = -\bGamma_0^{-1}\big(\widehat{\bU}^*_{n,l} + \widehat{\bU}_{n,l} \big) + o_{\P}(\rho_n^{-1}) \quad \text{for} \quad l \in\{0,\dots,L\}
\end{equation}
and that
\begin{equation}\label{eq: U-statistic asymptotic normality - bootstrap}
    \left[ n^{-1}\bSigma_0 + 3 \binom{n}{2}^{-1}h_n^{-d}\bDelta_0(\bar{K})\right]^{-1/2} \widetilde{\bU}^*_n \rightsquigarrow_\P \cN (\0_{k \times 1},\bI_k),
\end{equation}
where
\begin{equation*}
    \widetilde{\bU}^*_n = \sum_{l=0}^{L} \lambda_l(\bc)\widehat{\bU}^*_{n,l}.
\end{equation*}

For every $\bt \in \R^k$, using a Hoeffding decomposition and the fact that $m(\bz,\bz;\btheta_{n,l})=0$ and $\bs(\bz,\bz;\btheta_{n,l})=\0$ for large $n$, we have  
\begin{align*}
	&\rho_n^2\Big[\widehat{M}^*_n(\btheta_{n,l}+\bt \rho_n^{-1};h_{n,l}) - \widehat{M}^*_n(\btheta_{n,l};h_{n,l}) \Big] \\
    &= \rho_n^2 (1+o(1)) [\widehat{M}_n(\btheta_{n,l}+\bt \rho_n^{-1};h_{n,l}) - \widehat{M}_n(\btheta_{n,l};h_{n,l}) ] + \bt'\rho_n\widehat{\bU}^*_{n,l} + o_\P(1),
\end{align*}
where it can be shown that $\rho_n\widehat{\bU}^*_{n,l}=O_\P(1)$ and where it follows from the proof of \eqref{eq: Hjort-Pollard stochastic expansion} that
\begin{equation*}
    \rho_n^2 \Big[\widehat{M}_n(\btheta_{n,l}+\bt \rho_n^{-1};h_{n,l}) - \widehat{M}_n(\btheta_{n,l};h_{n,l}) \Big] = \frac{1}{2}\bt'\bGamma_0\bt + \bt'\rho_n\widehat{\bU}_{n,l} + o_\P(1),
\end{equation*}
where $\rho_n\widehat{\bU}_{n,l}=O_\P(1)$. In other words,
\begin{align*}
    &\rho_n^2\Big[\widehat{M}^*_n(\btheta_{n,l}+\bt \rho_n^{-1};h_{n,l}) - \widehat{M}^*_n(\btheta_{n,l};h_{n,l}) \Big] \\
    &= \frac{1}{2}\bt'\bGamma_0\bt + \bt'\rho_n\big(\widehat{\bU}^*_{n,l} + \widehat{\bU}_{n,l}\big) + o_\P(1) \qquad \text{for every } \bt \in \R^k.
\end{align*}
Because $\bGamma_0$ is positive definite and because $\bt \mapsto \widehat{M}^*_n(\btheta_{n,l}+\bt \rho_n^{-1};h_{n,l})$ is convex (almost surely),  the corollary following \citet[Lemma 2]{Hjort-Pollard_1993} implies that \eqref{eq: Hjort-Pollard stochastic expansion - bootstrap} holds. 

To prove \eqref{eq: U-statistic asymptotic normality - bootstrap}, we begin by decomposing $\widetilde{\bU}^*_n$ as
\begin{equation*}
    \widetilde{\bU}^*_n = \widetilde{\bL}^*_n + \widetilde{\bW}^*_n,
\end{equation*}
where
\begin{equation*}
    \widetilde{\bL}^*_n = \frac{2}{n} \sum_i \widetilde{\bell}^*_n(\bz^*_{i,n}), \qquad \widetilde{\bell}^*_n(\bz^*_{i,n})=\sum_{l=0}^{L} \lambda_l(\bc) \E^*_n[\bs^{\mu,*}_{n,l}(\bz^*_{i,n},\bz^*_{j,n})|\bz^*_{i,n}] \qquad (j \neq i),
\end{equation*}
and
\begin{equation*}
    \widetilde{\bW}^*_n = \binom{n}{2}^{-1} \sum_{i<j} \widetilde{\bomega}^*_n(\bz^*_{i,n},\bz^*_{j,n}), \qquad \widetilde{\bomega}^*_n(\bz^*_{i,n},\bz^*_{j,n})=\sum_{l=0}^{L} \lambda_l(\bc) \big[\bs^{\mu,*}_{n,l}(\bz^*_{i,n},\bz^*_{j,n}) - \widetilde{\bell}^*_n(\bz^*_{i,n}) - \widetilde{\bell}^*_n(\bz^*_{j,n})\big].
\end{equation*}
Defining
\begin{equation*}
    \pi_n = \frac{\sqrt{nh_n^d}}{1+\sqrt{nh_n^d}},
\end{equation*}
routine calculations can be used to show that for every $\bmu_1,\bmu_2 \in \R^k$, we have
\begin{equation*}
    \widehat{\sigma}_n^2 = \sum_i\V^*_n[g^*_{i,n}] = \bmu_1'\big[\pi_n^2\bSigma_0 + 4(1-\pi_n)^2\bDelta_0(\bar{K})\big]\bmu_1 + \bmu_2'\bDelta_0(\bar{K})\bmu_2 + o_\P(1), 
\end{equation*}
\begin{equation*}
    \sum_i \E^*_n[g^{*4}_{i,n}] = o_\P(1),
\end{equation*}
and
\begin{equation*}
    \V^*_n\left[\sum_i\widehat{\sigma}_{i,n}^2-\widehat{\sigma}_n^2\right] = o_\P(1), \qquad \widehat{\sigma}_{i,n}^2=\V^*_n[g^*_{i,n}|\bz^*_{1,n},\dots,\bz^*_{i-1,n}],
\end{equation*}
where $\V^*_n[\cdot]$ denotes $\V[\cdot|\bz_1,\dots,\bz_n]$ and where 
\begin{equation*}
    g^*_{i,n} = g^*_{i,n}(\bmu) = \frac{\pi_n}{\sqrt{n}} 2\bmu_1'\widetilde{\bell}^*_n(\bz^*_{i,n}) + \sqrt{\binom{n}{2}^{-1}h_n^{d}}\sum_{j=1}^{i-1}\bmu_2'\widetilde{\bomega}^*_{n}(\bz^*_{i,n},\bz^*_{j,n}).
\end{equation*}
The Cram{\'e}r-Wold device and the central limit theorem of \cite{Heyde-Brown_1970} therefore imply that if $\pi_n \to \pi_0 \in [0,1]$, then
\begin{equation*}
    \left(\begin{array}{c} \sqrt{n}\pi_n\widetilde{\bL}^*_n \\  \sqrt{\binom{n}{2}h_n^d} \widetilde{\bW}^*_n \end{array}\right) \rightsquigarrow_\P \cN\left( \left[\begin{array}{c} \0_{k \times 1} \\ \0_{k \times 1} \end{array}\right], \left[\begin{array}{cc} \pi_0^2\bSigma_0 + 4(1-\pi_0)^2\bDelta_0(\bar{K}) &\0_{k \times k} \\ \0_{k \times k} &\bDelta_0(\bar{K})  \end{array}\right] \right).
\end{equation*}
Whether or not $\pi_n$ is convergent, the result \eqref{eq: U-statistic asymptotic normality - bootstrap} can be obtained from the preceding display by arguing along subsequences (if necessary).

\subsection{Verifying Assumption \ref{Assumption: Bias of thetahat}}\label{Section: Verifying Assumption 3}

It follows from Lemma \ref{Lemma: Existence and Convergence of theta(h)} that if Assumption \ref{Assumption: Convergence of M} holds, then so does Assumption \ref{Assumption: Bias of thetahat} with $\smooth=0$. This observation provides the base case for an induction argument. To describe the induction step, suppose that for some even $S \geq 0$, we have
\begin{equation*}
    \btheta(h) -\btheta_0 = \sum_{l=1}^{S/2} \bb_{2l} h^{2l} + o(h^{S}) \qquad \text{ as } h\downarrow 0.
\end{equation*}
Suppose also that, for every $\bt \in \R^k$ and some $\bbeta_{S+2} \in \R^k$, we have
\begin{align*}
    &h^{-(S+2)} \left[ M \left(\btheta_0 + \sum_{l=1}^{S/2} \bb_{2l} h^{2l} + \bt h^{S+2};h \right) - M \left(\btheta_0 + \sum_{l=1}^{S/2} \bb_{2l} h^{2l}; h  \right)\right] \\
    &= \bt' \bbeta_{S+2} + \frac{1}{2} \bt' \bGamma_0 \bt + o(1) \qquad \text{as } h \downarrow 0.
\end{align*}
Then, the corollary following \citet[Lemma 2]{Hjort-Pollard_1993} implies that
\begin{align*}
    h^{-(S+2)} \left( \btheta(h) - \btheta_0 - \sum_{l=1}^{S/2} \bb_{2l} h^{2l} \right) &= \argmin_{\bt \in \R^k} M \left(\btheta_0 + \sum_{l=1}^{S/2} \bb_{2l} h^{2l} + \bt h^{S+2};h \right) \\
    &= - \bGamma_0^{-1} \bbeta_{S+2} + o(1) \qquad \text{ as } h\downarrow 0;
\end{align*}
that is, defining $\bb_{S+2} = - \bGamma_0^{-1} \bbeta_{S+2}$, we have
\begin{equation*}
    \btheta(h) -\btheta_0 = \sum_{l=1}^{(S+2)/2} \bb_{2l} h^{2l} + o(h^{S+2}) \qquad \text{ as } h\downarrow 0.
\end{equation*}

To complete the proof of Proposition \ref{Proposition: Bias Expansion L=2}, it therefore suffices to note that (for every $\bt \in \R^k$ and) under the assumptions of the proposition, we have
\begin{equation*}
    h^{-4} \left[ M (\btheta_0 + \bt h^2;h ) - M (\btheta_0 ; h )\right]
    = \bt' \bbeta_2 + \frac{1}{2} \bt' \bGamma_0 \bt + o(1) \qquad \text{as } h \downarrow 0,
\end{equation*}
where
\begin{equation*}
    \bbeta_2 = \frac{1}{2}\sum_{i=1}^d\sum_{j=1}^d
    \left(\int_\cW
    \frac{\partial^2 \bvarphi(\bw,\bar{\bw})}{\partial \bar w_i \partial \bar w_j}\Big\vert_{\bar{\bw}=\bw}
    f_{\bw}(\bw)d\bw\right)
    \left(\int_{\R^d} u_i u_j K(\bu)d\bu\right).
\end{equation*}

Similarly, if in addition to the assumptions of Proposition \ref{Proposition: Bias Expansion L=2} it is assumed that for every $\bt \in \R^k$ and for some $\bbeta_4 \in \R^k$, we have
\begin{equation*}
    h^{-8} \left[ M (\btheta_0 + \bb_2 h^2 + \bt h^4;h ) - M (\btheta_0 + \bb_2 h^2; h )\right] = \bt' \bbeta_4 + \frac{1}{2} \bt' \bGamma_0 \bt + o(1) \qquad \text{as } h \downarrow 0,
\end{equation*}
then Assumption \ref{Assumption: Bias of thetahat} holds with $\smooth=4$. One set of sufficient conditions for this to occur is that Assumptions \ref{Assumption: Convergence of M}-\ref{Assumption: Asymptotic Distribution} hold and that, for every $\bt \in \R^k$, the following are satisfied (with probability one):

\begin{enumerate}[(i)]
    \item $\int_{\R^d} \|\bu\|^4 K(\bu)d\bu<\infty$.

    \item $\bar{\bw} \mapsto \bvarphi(\bw,\bar{\bw}) = \E[\bs(\bz_1,\bz_2;\btheta_0)|\bw_1=\bw,\bw_2=\bar{\bw}]f_{\bw}(\bar{\bw})$ is four times continuously differentiable with $\E[\sup_{\bar{\bw}\in\cW}\|\partial_{\bar{\bw}}^{\bnu} \bvarphi(\bw,\bar{\bw})\|]<\infty$ for all $\bnu\in\Z_+^d$ with $|\bnu|\leq 4$.

    \item $f_\bw$ is twice continuously differentiable and $\bar{\bw} \mapsto \bH(\bw,\bar{\bw};\btheta_0,\bt)$ is twice continuously differentiable with $\E[\sup_{\bar{\bw}\in\cW}\|\partial_{\bar{\bw}}^{\bnu} \bH(\bw,\bar{\bw};\btheta_0,\bt)f_\bw(\bar{\bw})\|]<\infty$ for all $\bnu\in\Z_+^d$ with $|\bnu|\leq 2$.

    \item For some function $\dot\bH(\bw,\bar{\bw};\btheta,\bt)\in\R^{k\times k}$, $\bar{\bw} \mapsto \dot\bH(\bw,\bar{\bw};\btheta_0,\bt)$ is continuous,
    \begin{equation*}
        \E\left[\sup_{\bar{\bw}\in\cW}\|\dot \bH(\bw,\bar{\bw};\btheta_0,\bt)f_\bw(\bar{\bw})\|\right]<\infty,
    \end{equation*}
    \begin{equation*}
        \lim_{\tau\downarrow0,(\btheta,\bu)\to(\btheta_0,\0)} \left\| \frac{\bH(\bw,\bw+\bu;\btheta+\tau\bt,\bt)-\bH(\bw,\bw+\bu;\btheta,\bt)}{\tau}-\dot\bH(\bw,\bw+\bu;\btheta,\bt) \right\|= 0,
    \end{equation*}
    and, for some $\delta>0$,
    \begin{equation*}
        \E\left[\sup_{ \tau\in (0,\delta),\|\btheta-\btheta_0\|<\delta, \bw_2 \in\cW} \left\| \frac{\bH(\bw_1,\bw_2;\btheta+\tau\bt,\bt)-\bH(\bw_1,\bw_2;\btheta,\bt)}{\tau}-\dot\bH(\bw_1,\bw_2;\btheta,\bt) \right\|\right]< \infty.
    \end{equation*}
\end{enumerate}

\section{Sufficient Conditions for Motivating Examples}\label{Section: Sufficient Conditions for Motivating Examples}

To demonstrate the plausibility of Assumptions \ref{Assumption: Convergence of M} and \ref{Assumption: Asymptotic Distribution}, we revisit the examples of Section \ref{Section: Motivating Examples}. In each example, Assumption \ref{Assumption: Convergence of M}\ref{Assumption: Convergence of M - convexity} holds and Assumption \ref{Assumption: Convergence of M}\ref{Assumption: Convergence of M - density of w} is fairly primitive, so we focus on giving primitive sufficient conditions for Assumptions \ref{Assumption: Convergence of M}\ref{Assumption: Convergence of M - well defined Mn}-\ref{Assumption: Convergence of M - well behaved M0} and \ref{Assumption: Asymptotic Distribution}.

\subsection{Partially Linear Regression Model}
We take $\bs=\bs_{\mathtt{PLR}}$ and $\bH=\bH_{\mathtt{PLR}}$, where
\begin{equation*}
    \bs_{\mathtt{PLR}}(\bz_i,\bz_j;\btheta) = \frac{\partial}{\partial \btheta} m_{\mathtt{PLR}}(\bz_i,\bz_j;\btheta) = -\dot{\bx}_{i,j} (\dot{y}_{i,j}-\dot{\bx}_{i,j}'\btheta)
\end{equation*}
and
\begin{align*}
    \bH_{\mathtt{PLR}}(\bw_i,\bw_j)
    &= \frac{\partial^2}{\partial \btheta \partial \btheta'} \E[ m_{\mathtt{PLR}}(\bz_i,\bz_j;\btheta) | \bw_i,\bw_j]
     = \frac{\partial}{\partial \btheta'} \E[ \bs_{\mathtt{PLR}}(\bz_i,\bz_j;\btheta) | \bw_i,\bw_j] \\
    &= \E[\dot{\bx}_{i,j}\dot{\bx}_{i,j}' | \bw_i,\bw_j],
\end{align*}
the latter depending on neither $\btheta$ nor $\bt$ (because $m_{\mathtt{PLR}}(\bz_i,\bz_j;\btheta)$ is quadratic in $\btheta$).

Under mild conditions, Assumptions \ref{Assumption: Convergence of M}\ref{Assumption: Convergence of M - well defined Mn}-\ref{Assumption: Convergence of M - well behaved M0} and \ref{Assumption: Asymptotic Distribution} hold with
\begin{align*}
    \bxi_0(\bz)
    &= -2\E[\bs_{\mathtt{PLR}}(\bz_1,\bz_2;\btheta_0)|\bz_1=\bz,\bw_2=\bw ] f_{\bw}(\bw) = 2 (\bx-\E[\bx|\bw]) f_{\bw}(\bw)\varepsilon,
\end{align*}
\begin{align*}
    \bXi_0(\bw)
    &= \E[\bs_{\mathtt{PLR}}(\bz_1,\bz_2;\btheta_0)\bs_{\mathtt{PLR}}(\bz_1,\bz_2;\btheta_0)'|\bw_1=\bw,\bw_2=\bw ] f_{\bw}(\bw) \\
    &= \E[ \dot{\bx}_{1,2}\dot{\bx}_{1,2}'(\varepsilon_1^2+\varepsilon_2^2)|\bw_1=\bw,\bw_2=\bw ] f_{\bw}(\bw),
\end{align*}
and
\begin{align*}
    \bG_0(\bw) &= \bH_{\mathtt{PLR}}(\bw,\bw) f_{\bw}(\bw) = 2 \V[\bx|\bw] f_{\bw}(\bw).
\end{align*}

For instance, it suffices to set $b(\bz) = (1+\|\bx\|) (1+\|\bx\|+|\gamma_0(\bw)|+|\varepsilon|)$ and to assume that
\begin{enumerate}[(i)]
    \item The functions $\bw\mapsto \gamma_0(\bw)$, $\bw\mapsto \E[\bx|\bw]$, $\bw\mapsto \E[\bx\bx'|\bw]$, $\bw\mapsto \E[\varepsilon^2|\bw]$, $\bw\mapsto \E[\bx\varepsilon^2|\bw]$, and $\bw\mapsto \E[\bx\bx'\varepsilon^2|\bw]$ are continuous on $\cW$.

    \item $\E[(1+\|\bx\|^4)\varepsilon^4] + \sup_{\bw\in\cW}\E[(1+\|\bx\|^4)\varepsilon^4 |\bw ] f_{\bw}(\bw)<\infty$ and
    \begin{equation*}
        \E[(1+\|\bx\|^4)\gamma_0(\bw)^4 + \|\bx\|^8] + \sup_{\bw\in\cW}\E[(1+\|\bx\|^4)\gamma_0(\bw)^4 + \|\bx\|^8|\bw]f_{\bw}(\bw)<\infty.
    \end{equation*}

    \item With probability one, $\V[\bx|\bw]$ is positive definite and $\V[\varepsilon|\bx,\bw]>0$.
\end{enumerate}

Under the additional assumptions of Theorem \ref{Theorem: Asymptotic Distribution} and if $nh_n^d \to \infty$, the pairwise difference estimator is asymptotically linear with influence function
\begin{equation*}
    \bz \mapsto \E[\V[\bx|\bw]f_{\bw}(\bw)]^{-1} (\bx-\E[\bx|\bw]) f_{\bw}(\bw)\varepsilon.
\end{equation*}
Unless the distribution of $\bw$ is uniform on $\cW$, the pairwise difference estimator is therefore asymptotically distinct from the estimators studied by \cite{Robinson_1988_ECMA} and \cite{Newey-Donald_1994_JMA}, the influence function of these estimators being given by
\begin{equation*}
    \bz \mapsto \E[\V[\bx|\bw]]^{-1} (\bx-\E[\bx|\bw]) \varepsilon.
\end{equation*}

When $nh_n^d \to 2c<\infty$, the pairwise difference estimator is still $\sqrt{n}$-normal, but it ceases to be asymptotically linear. Similarly, the estimator of \cite{Newey-Donald_1994_JMA} is $\sqrt{n}$-normal, but not asymptotically linear, when the associated tuning parameter is chosen appropriately \citep{Cattaneo-Jansson-Newey_2018_ET}, and in light of \cite{Linton_1995_ECMA} it stands to reason that the same is true for the estimator of \cite{Robinson_1988_ECMA}. Another $\sqrt{n}$-normal estimator that is not asymptotically linear was proposed by \cite{Yatchew_1997_EL}. However, even under simplifying assumptions (e.g., uniformity of the distribution of $\bw$ and/or conditional homoskedasticity of $\varepsilon$), it appears difficult to make insightful comparisons between the various estimators just mentioned.

\subsection{Partially Linear Logit Model}
We take $\bs=\bs_{\mathtt{PLL}}$ and $\bH=\bH_{\mathtt{PLL}}$, where
\begin{equation*}
    \bs_{\mathtt{PLL}}(\bz_i,\bz_j;\btheta) = \frac{\partial}{\partial \btheta} m_{\mathtt{PLL}}(\bz_i,\bz_j;\btheta) = -\dot{\bx}_{i,j} (y_i - \Lambda(\dot{\bx}_{i,j}'\btheta))\1\{\dot{y}_{i,j}\neq0\}
\end{equation*}
and, defining $\lambda(u) = \partial \Lambda(u)/\partial u = \exp(u)/[1+\exp(u)]^2$,
\begin{align*}
    \bH_{\mathtt{PLL}}(\bw_i,\bw_j;\btheta) &= \frac{\partial^2}{\partial \btheta \partial \btheta'} \E[ m_{\mathtt{PLL}}(\bz_i,\bz_j;\btheta) | \bw_i,\bw_j] = \frac{\partial}{\partial \btheta'} \E[ \bs_{\mathtt{PLL}}(\bz_i,\bz_j;\btheta) | \bw_i,\bw_j] \\
    &= \E[ \dot{\bx}_{i,j}\dot{\bx}_{i,j}' \lambda(\dot{\bx}_{i,j}'\btheta)\1\{\dot{y}_{i,j}\neq0\}| \bw_i,\bw_j],
\end{align*}
where the latter does not depend on $\bt$ (because $m_{\mathtt{PLL}}(\bz_i,\bz_j;\btheta)$ is twice differentiable in $\btheta$).

Under mild conditions, Assumptions \ref{Assumption: Convergence of M}\ref{Assumption: Convergence of M - well defined Mn}-\ref{Assumption: Convergence of M - well behaved M0} and \ref{Assumption: Asymptotic Distribution} hold with
\begin{align*}
    \bxi_0(\bz)
    &= -2\E[\bs_{\mathtt{PLL}}(\bz_1,\bz_2;\btheta_0)|\bz_1=\bz,\bw_2=\bw ] f_{\bw}(\bw),
%    &= 2\E\left[ \left. \dot{\bx}_{1,2} (y_1-\Lambda(\dot{\bx}_{1,2}'\btheta_0)) \1\{\dot{y}_{1,2}\neq 0\}  \right\vert \bz_1=\bz,\bw_2=\bw \right] f_{\bw}(\bw),
\end{align*}
\begin{align*}
    \bXi_0(\bw)
    &= \E[\bs_{\mathtt{PLL}}(\bz_1,\bz_2;\btheta_0)\bs_{\mathtt{PLL}}(\bz_1,\bz_2;\btheta_0)'|\bw_1=\bw,\bw_2=\bw ] f_{\bw}(\bw),
%    &= \E\left[\left. (y_1-\Lambda(\dot{\bx}_{1,2}'\btheta_0) )^2\dot{\bx}_{1,2}\dot{\bx}_{1,2}'\1\{\dot{y}_{1,2}\neq 0\} \right|\bw_1=\bw,\bw_2=\bw \right] f_{\bw}(\bw),
\end{align*}
and
\begin{align*}
    \bG_0(\bw) &= \bH_{\mathtt{PLL}}(\bw,\bw;\btheta_0) f_{\bw}(\bw).
%    &= \E\left[\left. \dot{\bx}_{1,2}\dot{\bx}_{1,2}' \lambda(\dot{\bx}_{1,2}'\btheta_0)\1\{\dot{y}_{1,2}\neq 0\}\right\vert \bw_1=\bw,\bw_2=\bw\right]f_{\bw}(\bw).
\end{align*}

For instance, it suffices to set $b(\bz) = 1+\|\bx\|$ and to assume that, for some $\delta>0$,
\begin{enumerate}[(i)]
    \item The function $\bw\mapsto \gamma_0(\bw)$ is continuous on $\cW$. Also, the conditional distribution of $\bx$ given $\bw$ admits a density $f_{\bx|\bw}$ with respect to some measure $\rho$ such that $\bw\mapsto f_{\bx|\bw}(\bx|\bw)$ is continuous on $\cW$ (with probability one) and
    \begin{equation*}
        \int_{\R^k} (1+ \|\bx\|^2) \sup_{\|\bu\|\leq\delta} f_{\bx|\bw}(\bx| \bw + \bu)d\rho(\bx)<\infty \qquad \text{for every } \bw\in\cW. 
    \end{equation*}
    
    \item $\E[\|\bx\|^4] + \sup_{\bw\in\cW}\E[\|\bx\|^4 |\bw]f_{\bw}(\bw)<\infty.$

    \item With probability one, $\V[\bx|\bw]$ is positive definite.
\end{enumerate}

\subsection{Partially Linear Tobit Model}
We take $\bs=\bs_{\mathtt{PLT}}$ and $\bH=\bH_{\mathtt{PLT}}$, where
\begin{equation*}
    \bs_{\mathtt{PLT}}(\bz_i,\bz_j;\btheta) = -\dot{\bx}_{i,j} \left(\1\{y_i>\max(y_j+\dot{\bx}_{i,j}'\btheta,0)\} - \1\{y_j>\max(y_i-\dot{\bx}_{i,j}'\btheta,0)\} \right)
\end{equation*}
and
\begin{align*}
    \bH_{\mathtt{PLT}}(\bw_i,\bw_j;\btheta,\bt)
    &= \E[ \dot{\bx}_{i,j}\dot{\bx}_{i,j}' \left(\1\{\dot{\bx}_{i,j}'\btheta>0\}+\1\{\dot{\bx}_{i,j}'\btheta=0,\dot{\bx}_{i,j}'\bt\geq0\}\right)\eta_{\mathtt{PLT}}(\bz_i,\bz_j;\btheta)| \bw_i,\bw_j]  \\
    &\quad + \E[ \dot{\bx}_{i,j}\dot{\bx}_{i,j}' \left(\1\{\dot{\bx}_{i,j}'\btheta<0\}+\1\{\dot{\bx}_{i,j}'\btheta=0,\dot{\bx}_{i,j}'\bt<0\}\right)\eta_{\mathtt{PLT}}(\bz_j,\bz_i;\btheta)| \bw_i,\bw_j],
\end{align*}
with
\begin{align*}
    \eta_{\mathtt{PLT}}(\bz_i,\bz_j;\btheta) &= 2\int_0^{\infty} f_{\varepsilon|\bw}(\varepsilon-\bx_i'\btheta_0-\gamma_0(\bw_i)+\dot{\bx}_{i,j}'\btheta|\bw_i) f_{\varepsilon|\bw}(\varepsilon-\bx_j'\btheta_0-\gamma_0(\bw_j)|\bw_j)d\varepsilon  \\
    &\quad + f_{\varepsilon|\bw}(-\bx_i'\btheta_0-\gamma_0(\bw_i)+\dot{\bx}_{i,j}'\btheta|\bw_i) \int_{-\infty}^0  f_{\varepsilon|\bw}(\varepsilon-\bx_j'\btheta_0-\gamma_0(\bw_j)|\bw_j)d\varepsilon.
\end{align*}

Under mild conditions, Assumptions \ref{Assumption: Convergence of M}\ref{Assumption: Convergence of M - well defined Mn}-\ref{Assumption: Convergence of M - well behaved M0} and \ref{Assumption: Asymptotic Distribution} hold with
\begin{align*}
    \bxi_0(\bz)
    &= -2\E[\bs_{\mathtt{PLT}}(\bz_1,\bz_2;\btheta_0)|\bz_1=\bz,\bw_2=\bw ] f_{\bw}(\bw),
%    &=2\E\big[ \dot{\bx}_{1,2} \big(1-F_{y|\bx,\bw}(\max(y_1-\dot{\bx}_{1,2}'\btheta_0,0)|\bx_2,\bw ) \\
%    &\qquad - \1\{y_1>0\}F_{y|\bx,\bw}(y_1-\dot{\bx}_{1,2}'\btheta_0|\bx_2,\bw) \big)\vert \bz_1=\bz,\bw_2=\bw  \big]f_{\bw}(\bw),
\end{align*}
\begin{align*}
    \bXi_0(\bw)
    &= \E[\bs_{\mathtt{PLT}}(\bz_1,\bz_2;\btheta_0)\bs_{\mathtt{PLT}}(\bz_1,\bz_2;\btheta_0)'|\bw_1=\bw,\bw_2=\bw ] f_{\bw}(\bw),
%    &= \E\Big[ \dot{\bx}_{1,2}\dot{\bx}_{1,2}'\big( \P[y_2>0,y_1-y_2<\dot{\bx}_{1,2}'\btheta_0|\bx_1,\bx_2,\bw_1,\bw_2]\\
%    &\qquad\quad  +\P[y_1>0,y_1-y_2>\dot{\bx}_{1,2}'\btheta_0|\bx_1,\bx_2,\bw_1,\bw_2]\big) \Big\vert \bw_1=\bw,\bw_2=\bw\Big]f_{\bw}(\bw),
\end{align*}
and
\begin{align*}
    \bG_0(\bw) &= \bH_{\mathtt{PLT}}(\bw,\bw;\btheta_0,\btheta_0) f_{\bw}(\bw).
%    &= \E\bigg[\dot{\bx}_{1,2}\dot{\bx}_{1,2}'\bigg(2\int_0^{\infty} f_{y^*|\bx,\bw}(y+\max\{0,\dot{\bx}_{1,2}'\btheta_0\}|\bx_2,\bw)f_{y^*|\bx,\bw}(y-\min\{0,\dot{\bx}_{1,2}'\btheta_0\}|\bx_1,\bw)dy\\
%    &\qquad + \1\{\dot{\bx}_{1,2}'\btheta_0 \geq 0 \}f_{y^*|\bx,\bw}(0|\bx_2,\bw)F_{y^*|\bx,\bw}(0|\bx_2,\bw) \\
%    &\qquad + \1\{\dot{\bx}_{1,2}'\btheta_0<0\}f_{y^*|\bx,\bw}(0|\bx_1,\bw)F_{y^*|\bx,\bw}(0|\bx_1,\bw)\bigg)\bigg\vert\bw_1=\bw,\bw_2=\bw \bigg]f_{\bw}(\bw)
\end{align*}

For instance, it suffices to set $b(\bz) = 1+\|\bx\|$ and to assume that, for some $\delta>0$,
\begin{enumerate}[(i)]
    \item The function $\bw\mapsto \gamma_0(\bw)$ is continuous on $\cW$. Also, the conditional distribution of $\bx$ given $\bw$ admits a density $f_{\bx|\bw}$ with respect to some measure $\rho$ such that $\bw\mapsto f_{\bx|\bw}(\bx|\bw)$ is continuous on $\cW$ (with probability one) and
    \begin{equation*}
        \int_{\R^k} (1+ \|\bx\|^2) \sup_{\|\bu\|\leq\delta} f_{\bx|\bw}(\bx| \bw + \bu)d\rho(\bx)<\infty \qquad \text{for every } \bw\in\cW. 
    \end{equation*}
    In addition, the function $(\varepsilon,\bw)\mapsto f_{\varepsilon|\bw}(\varepsilon|\bw)$ is continuous and bounded and the function
    \begin{equation*}
        (\bx,\bw)\mapsto \int_\R \sup_{|u|+\|\bu\|\leq\delta}f_{\varepsilon|\bw}(\varepsilon-\bx'\btheta_0-\gamma_0(\bw)+u|\bw+\bu)d\varepsilon
    \end{equation*}
    is bounded.
    
    \item $\E[\|\bx\|^4] + \sup_{\bw\in\cW}(1+\E[\|\bx\|^4 |\bw])f_{\bw}(\bw)<\infty.$

    \item With probability one, $\V[\bx|\bw]$ is positive definite.
\end{enumerate}

\section{Simulation Evidence}\label{Section: Simulation Evidence}

We present simulation evidence for the partially linear regression and partially linear logit models. We compare the four bootstrap-based confidence intervals discussed in Section \ref{Section: Discussion}: $\mathsf{CI}_{n,1-\alpha}^*(1,0)$, $\mathsf{CI}_{n,1-\alpha}^*(1,L)$, $\mathsf{CI}_{n,1-\alpha}^*(3^{1/d},0)$, and $\mathsf{CI}_{n,1-\alpha}^*(3^{1/d},L)$. We set $\alpha=0.05$, and for the debiasing procedure, we use $L=1$ and $\bc=(1,2)'$.

\subsection{Simulation Design: Partially Linear Regression Model}

We base the partially linear regression designs on equations (6.1)--(6.2) of \cite{Robinson_1988_ECMA}. We consider
\begin{equation*}
    y_i = x_i + \gamma_0(\bw_i) + \varepsilon_i,
\end{equation*}
The three designs are as follows. In Model 1, $x_i=2v_i$, $v_i\sim\mathrm{N}(0,1)$, $w_i=v_i+u_i$, $u_i\sim \mathrm{N}(0,2)$, $\gamma_0(w)=w^2+1$, $\varepsilon_i\sim\mathrm{N}(0,1)$, and $\{v_i,u_i,\varepsilon_i\}$ are jointly independent. In Model 2, $d=\dim(\bw)=2$, $\gamma_0(\bw)=\bw'\bw$, $(x_i,\bw_i')'$ is normal with each component having mean $1$ and variance $3$, and any pair in $(x_i,\bw_i')'$ has correlation $2/3$. Model 3 uses the same joint normal specification for $(x_i,\bw_i')'$ as Model 2, but with $d=\dim(\bw)=3$ and $\gamma_0(\bw)=\bw'\bw-(1+\dim(\bw))$.
In all models, $\theta_0=1$.
The regressor of interest $x_i$ is correlated with $\bw_i$, and thus, ignoring the $\gamma_0(\bw_i)$ term will induce bias in the estimator of $\theta_0$.

\subsection{Simulation Design: Partially Linear Logit Model}

We base the partially linear logit designs on those of \cite{Honore-Powell_2005_Festschrift}. As in the partially linear regression simulations, the three designs differ mainly in the dimension of $\bw$. We consider
\begin{equation*}
    y_i = \1\big\{ x_{1i} + x_{2i} + \gamma_0(\bw_i) + \varepsilon_i\geq 0 \big\},\qquad \gamma_0(\bw) = \bw'\bw-(1+\dim(\bw)),
\end{equation*}
where the CDF of $\varepsilon_i$ is $\Lambda(u)=1/(1+\exp(-u))$, $x_{2i}$ has a discrete distribution with $\P[x_{2i}=1]=1/2=\P[x_{2i}=-1]$, $x_{1i}=v_i + \bw_i'\bw_i$ with $v_i\sim \mathrm{N}(0,1)$, and $\{\varepsilon_i,\bw_i,x_{2i},v_i\}$ are jointly independent. 
For Model $d$ ($d\in\{1,2,3\}$), $\dim(\bw)=d$, and $\bw_i$ has a normal distribution with each element having mean zero and variance one, equicorrelated with correlation $0.2$.
In all models, $\btheta_0=(1, 1)'$.

As noted by \cite{Honore-Powell_2005_Festschrift}, ignoring the presence of $\gamma_0(\bw_i)$ induces bias in the estimator of $\btheta_0$, although the bias for the second element tends to be negligible relative to the standard error. We focus on constructing $95\%$ confidence intervals for the first element of $\btheta_0$. 

\subsection{Simulation Results}

Tables~\ref{Table: plr d=1}--\ref{Table: plr d=3} display the coverage
probabilities and average lengths of the four confidence intervals for the
partially linear regression model, and Tables~\ref{Table: logit d=1}--\ref{Table: logit d=3}
for the partially linear logit model. Figure~\ref{fig:coverage} complements the
tables by plotting coverage probability as a function of bandwidth $h$ for all six DGPs. For each
data generating process, the sample size was $n=2{,}000$, we conducted $2{,}000$ simulation replications, and
for each replication we drew $2{,}000$ bootstrap samples to compute bootstrap quantiles. 

The panel headed ``$B=1$'' refers to $\mathsf{CI}_{n,0.95}^*(1,L)$, in which the bandwidth used to compute
bootstrapped quantiles equals the bandwidth used for estimation. The panel headed
``$B=3^{1/d}$'' refers to our proposed methods $\mathsf{CI}_{n,0.95}^*(3^{1/d},L)$, in which the bandwidth for bootstrap estimation is
rescaled by the factor $3^{1/d}$. Within each panel, $L=0$ is based on
the estimator $\widehat{\btheta}_n$ and $L=1$ is based on the debiased
estimator $\widetilde{\btheta}_n$ with $\bc=(1, 2)'$.

To form the bandwidth grid, we first estimate $h_{L=0}$ and $h_{L=1}$ by
auxiliary simulations. These are the MSE-minimizing bandwidths for
$\widehat{\btheta}_n$ and $\widetilde{\btheta}_n$, respectively. The grid is
then taken as the union of $h_{L=0}\times\{0.5,\dots,1.5\}$ and
$h_{L=1}\times\{0.5,\dots,1.5\}$. Rows highlighted in yellow and
orange in the tables mark $h_{L=0}$ and $h_{L=1}$, respectively. 

The results are broadly consistent with our theoretical predictions. We focus on
two forces of particular interest: small bandwidth regions, where the standard
bootstrap variance may be inflated, and large bandwidth regions, where smoothing
bias may become non-negligible. The simulations also include one instructive
qualification: in the partially linear regression model with $d=1$, the
conventional procedures are already very well calibrated, while our proposed
procedure $\mathsf{CI}_{n,0.95}^*(3^{1/d},1)$ is conservative. This is a
specialized design in which the smoothing bias targeted by jackknife debiasing is
negligible, so bandwidth rescaling and debiasing mostly widen the intervals without
offsetting a meaningful bias. This qualification is specific to the PLR $d=1$
design and does not describe the partially linear logit model with $d=1$, where
uncorrected intervals exhibit undercoverage as the bandwidth grows and bias
correction is beneficial.

\paragraph{Small-bandwidth overcoverage.}
For small bandwidths, the confidence intervals $\mathsf{CI}_{n,0.95}^*(1,L)$ tend to overcover the nominal $95\%$
level, whereas our proposed methods $\mathsf{CI}_{n,0.95}^*(3^{1/d},L)$ achieve coverage close to the nominal
probability.  This observation is predicted by our
small-bandwidth asymptotics: using the same $h_n$ for both estimation
and the bootstrap quantile computation causes the bootstrap to overestimate
the sampling dispersion of $\widehat{\btheta}_n$.  The average interval lengths in the tables are in line with this prediction.  
At small bandwidths, $\mathsf{CI}_{n,0.95}^*(1,L)$ are systematically longer
than their $\mathsf{CI}_{n,0.95}^*(3^{1/d},L)$ counterparts.

The severity of this distortion grows sharply with $d=\dim(\bw)$, consistent
with our theoretical prediction that the larger $d$, the wider the small-bandwidth region (based on the condition $\liminf_{n\to\infty}nh_n^d<\infty$).  
For the partially linear regression model with $d=2$ and $d=3$,
the coverage probabilities for $\mathsf{CI}_{n,0.95}^*(1,L)$ are well above the nominal level for smaller bandwidths,
whereas our methods $\mathsf{CI}_{n,0.95}^*(3^{1/d},L)$ achieve the nominal level (Tables~\ref{Table: plr d=2} and \ref{Table: plr d=3}). 
The same small-bandwidth variance distortion is visible in the logit model.
Overcoverage under $\mathsf{CI}_{n,0.95}^*(1,0)$ is modest for $d=1$, rises to $0.982$ for $d=2$,
and reaches $0.996$ for $d=3$ at the smallest bandwidth used in the simulation (Tables~\ref{Table: logit d=1}-\ref{Table: logit d=3}),
with analogous patterns for $\mathsf{CI}_{n,0.95}^*(1,1)$.
Figure~\ref{fig:coverage} makes this dimension-dependence visually apparent:
the coverage curves are well above the nominal level (the horizontal dotted line) for $d>1$.

\paragraph{Large-bandwidth undercoverage and bias correction.}
At larger bandwidths, except in the specialized PLR $d=1$ design discussed above, the bias may become non-negligible,
and $\mathsf{CI}_{n,0.95}^*(B,0)$ tends to exhibit severe undercoverage.  For
the partially linear regression model with $d=2$, the coverage of
$\mathsf{CI}_{n,0.95}^*(1,0)$ falls to $0.468$ at
$h=0.75$ and to $0.007$ at $h=1.12$.  For $d=3$, it falls to $0.051$
at $h=1.00$ and to essentially zero for $h\geq1.20$.  The undercoverage
is similarly pronounced in the logit model: for $d=2$, coverage collapses to
$0.858$ at $h=0.68$, and for $d=3$, it falls to $0.720$ at $h=0.75$.

The debiased intervals correct this distortion effectively. At larger bandwidths,
the $(B=1,L=1)$ procedure is closer to the nominal $95\%$ level than the
$(B=1,L=0)$ procedure: for the partially linear regression model with $d=2$,
coverage is $0.965$ at $h_{L=1}=0.75$, and for $d=3$ it is $0.970$ at
$h_{L=1}=1.00$. Similar patterns hold for the logit model, where coverage is
$0.963$ at $h_{L=1}=0.45$ for $d=2$ and $0.976$ at $h_{L=1}=0.50$ for $d=3$.

\paragraph{Combining bandwidth rescaling and bias correction.}

Under the procedure that combines bandwidth rescaling and bias correction, $\mathsf{CI}_{n,0.95}^*(3^{1/d},1)$,
coverage remains close to the nominal level across a wide bandwidth range.
For the partially linear regression model with $d=2$, 
the coverage probabilities lie between $0.949$ and $0.959$ across all bandwidths on the grid.
For $d=3$, coverage ranges from $0.950$ to $0.970$ over most of the grid and
remains at $0.955$ even at $h=1.50$.  The cost of bias correction is a modest
increase in average interval length relative to the uncorrected $L=0$
intervals: comparing the $(B=3^{1/d},L=0)$ and $(B=3^{1/d},L=1)$ columns at
the MSE-optimal bandwidth $h_{L=0}$, the debiased intervals are roughly
$4\%$ wider for $d=2$ and roughly $11\%$ wider for $d=3$ in the partially linear regression
model, a small price given the improvement in coverage reliability for larger bandwidths.
For the logit model, the qualitative picture is the same for $d=2$ and $d=3$.
For $d=1$, the proposed method is not conservative in the same way as in the PLR
case; instead, the main issue is the emergence of bias as the bandwidth grows,
which bias correction partly offsets. Under the proposed procedure, coverage
remains close to the nominal $95\%$ across a wide range of bandwidths considered.

Overall, the simulation evidence supports combining bandwidth rescaling ($B=3^{1/d}$) with jackknife-based bias correction ($L=1$), especially when smoothing bias is present or when the dimension of the nonparametric component is moderate.
The PLR $d=1$ design provides a useful reminder that when conventional methods are already well calibrated and the relevant smoothing bias is essentially absent, the robust procedure can be conservative.

\section{Conclusion}\label{Section: Conclusion}

This paper develops bandwidth-robust distribution theory and bootstrap-based inference procedures for a broad class of convex pairwise difference estimators. Our theoretical work is based on small bandwidth asymptotics and carefully leverages convexity. We illustrate the theory with three prominent examples. In addition to expanding the scope of small bandwidth asymptotics, our results lay the groundwork for several promising avenues of future research. First, our methods could be generalized to develop bandwidth selection based on higher-order stochastic expansions. Second, consistent variance estimators could be developed as an alternative to bootstrap-based inference. Third, it would be of interest to expand our theory to allow for pairwise difference estimators based on generated regressors, a class of estimators that sometimes arises in the context of control function and related econometric methods. Finally, it seems plausible that there exist settings where the objective function is sufficiently non-smooth to result in non-Gaussian distributional approximations. We leave these extensions for future work.

\begin{table}[hbtp]\renewcommand{\arraystretch}{1.2}
	\caption{Bootstrap 95\% Confidence Intervals for Partially Linear Regression Model: DGP 1 ($d=1$).}
	\label{Table: plr d=1}
	{\resizebox{\columnwidth}{!}{\begin{tabular}{r cccc cccc}
\toprule
& \multicolumn{4}{c}{$B = 1$} & \multicolumn{4}{c}{$B = 3^{1/d}$} \\
\cmidrule(lr){2-5}\cmidrule(lr){6-9}
& \multicolumn{2}{c}{$L = 0$} & \multicolumn{2}{c}{$L = 1$} & \multicolumn{2}{c}{$L = 0$} & \multicolumn{2}{c}{$L = 1$} \\
\cmidrule(lr){2-3}\cmidrule(lr){4-5}\cmidrule(lr){6-7}\cmidrule(lr){8-9}
$h$ & Coverage & Length & Coverage & Length & Coverage & Length & Coverage & Length \\
\midrule
0.35 & 0.958 & 0.058 & 0.958 & 0.059 & 0.956 & 0.058 & 0.956 & 0.058 \\
0.42 & 0.958 & 0.058 & 0.958 & 0.058 & 0.958 & 0.059 & 0.955 & 0.058 \\
0.49 & 0.956 & 0.058 & 0.958 & 0.058 & 0.963 & 0.060 & 0.955 & 0.058 \\
0.50 & 0.956 & 0.058 & 0.958 & 0.058 & 0.963 & 0.060 & 0.956 & 0.058 \\
0.56 & 0.955 & 0.058 & 0.957 & 0.058 & 0.966 & 0.061 & 0.958 & 0.058 \\
0.60 & 0.954 & 0.058 & 0.956 & 0.058 & 0.969 & 0.063 & 0.959 & 0.059 \\
0.63 & 0.953 & 0.058 & 0.956 & 0.058 & 0.970 & 0.064 & 0.961 & 0.059 \\
\rowcolor{yellow}0.70 & 0.954 & 0.058 & 0.956 & 0.058 & 0.977 & 0.067 & 0.965 & 0.061 \\
0.77 & 0.955 & 0.058 & 0.956 & 0.058 & 0.986 & 0.070 & 0.970 & 0.063 \\
0.80 & 0.956 & 0.058 & 0.955 & 0.058 & 0.987 & 0.072 & 0.972 & 0.064 \\
0.84 & 0.956 & 0.058 & 0.956 & 0.058 & 0.990 & 0.074 & 0.976 & 0.066 \\
0.90 & 0.954 & 0.058 & 0.955 & 0.058 & 0.993 & 0.078 & 0.982 & 0.069 \\
0.91 & 0.954 & 0.058 & 0.955 & 0.058 & 0.993 & 0.078 & 0.984 & 0.070 \\
0.98 & 0.953 & 0.058 & 0.953 & 0.058 & 0.997 & 0.083 & 0.987 & 0.074 \\
\rowcolor{orange} 1.00 & 0.954 & 0.058 & 0.953 & 0.058 & 0.998 & 0.084 & 0.989 & 0.075 \\
1.05 & 0.955 & 0.058 & 0.952 & 0.058 & 0.999 & 0.088 & 0.992 & 0.078 \\
1.10 & 0.955 & 0.058 & 0.952 & 0.058 & 1.000 & 0.091 & 0.994 & 0.081 \\
1.20 & 0.955 & 0.058 & 0.953 & 0.058 & 1.000 & 0.098 & 0.999 & 0.088 \\
1.30 & 0.953 & 0.059 & 0.954 & 0.058 & 1.000 & 0.105 & 1.000 & 0.095 \\
1.40 & 0.955 & 0.059 & 0.954 & 0.058 & 1.000 & 0.112 & 1.000 & 0.102 \\
1.50 & 0.957 & 0.060 & 0.954 & 0.058 & 1.000 & 0.119 & 1.000 & 0.108 \\
\bottomrule
\end{tabular}
}}
    \footnotesize\hspace{.5in}\vspace{0in}\\
	\raggedright Note: The table shows coverage probability and average length of confidence intervals over $2{,}000$ simulation replications. The sample size is $n=2{,}000$ and the number of bootstrap draws is $2{,}000$. The rows highlighted in yellow and orange correspond to the estimated MSE-optimal bandwidths for the point estimators with $L=0$ and $L=1$, respectively.
\end{table}

\begin{table}[hbtp]\renewcommand{\arraystretch}{1.2}
	\caption{Bootstrap 95\% Confidence Intervals for Partially Linear Regression Model: DGP 2 ($d=2$).}
	\label{Table: plr d=2}
	{\resizebox{\columnwidth}{!}{\begin{tabular}{r cccc cccc}
\toprule
& \multicolumn{4}{c}{$B = 1$} & \multicolumn{4}{c}{$B = 3^{1/d}$} \\
\cmidrule(lr){2-5}\cmidrule(lr){6-9}
& \multicolumn{2}{c}{$L = 0$} & \multicolumn{2}{c}{$L = 1$} & \multicolumn{2}{c}{$L = 0$} & \multicolumn{2}{c}{$L = 1$} \\
\cmidrule(lr){2-3}\cmidrule(lr){4-5}\cmidrule(lr){6-7}\cmidrule(lr){8-9}
$h$ & Coverage & Length & Coverage & Length & Coverage & Length & Coverage & Length \\
\midrule
0.12 & 0.991 & 0.149 & 0.994 & 0.171 & 0.949 & 0.111 & 0.949 & 0.122 \\
0.15 & 0.987 & 0.133 & 0.992 & 0.151 & 0.950 & 0.104 & 0.950 & 0.112 \\
0.17 & 0.984 & 0.123 & 0.990 & 0.137 & 0.951 & 0.100 & 0.955 & 0.106 \\
0.20 & 0.979 & 0.115 & 0.986 & 0.127 & 0.946 & 0.097 & 0.957 & 0.102 \\
0.22 & 0.977 & 0.110 & 0.984 & 0.120 & 0.944 & 0.095 & 0.955 & 0.099 \\
\rowcolor{yellow}  0.25 & 0.970 & 0.106 & 0.983 & 0.114 & 0.944 & 0.093 & 0.956 & 0.097 \\
0.28 & 0.965 & 0.102 & 0.979 & 0.110 & 0.939 & 0.092 & 0.956 & 0.095 \\
0.30 & 0.961 & 0.100 & 0.978 & 0.107 & 0.931 & 0.091 & 0.956 & 0.094 \\
0.32 & 0.952 & 0.098 & 0.978 & 0.104 & 0.923 & 0.090 & 0.956 & 0.093 \\
0.35 & 0.941 & 0.096 & 0.972 & 0.102 & 0.912 & 0.090 & 0.956 & 0.092 \\
0.38 & 0.925 & 0.095 & 0.973 & 0.100 & 0.905 & 0.089 & 0.956 & 0.091 \\
0.45 & 0.887 & 0.092 & 0.973 & 0.096 & 0.872 & 0.089 & 0.957 & 0.090 \\
0.52 & 0.829 & 0.091 & 0.968 & 0.094 & 0.821 & 0.089 & 0.958 & 0.089 \\
0.60 & 0.743 & 0.090 & 0.966 & 0.092 & 0.749 & 0.090 & 0.959 & 0.089 \\
0.68 & 0.625 & 0.089 & 0.964 & 0.091 & 0.646 & 0.091 & 0.958 & 0.088 \\
\rowcolor{orange}  0.75 & 0.468 & 0.089 & 0.965 & 0.090 & 0.514 & 0.093 & 0.957 & 0.089 \\
0.83 & 0.298 & 0.089 & 0.961 & 0.090 & 0.375 & 0.097 & 0.958 & 0.089 \\
0.90 & 0.151 & 0.089 & 0.952 & 0.089 & 0.239 & 0.100 & 0.956 & 0.090 \\
0.98 & 0.071 & 0.089 & 0.942 & 0.089 & 0.124 & 0.105 & 0.954 & 0.092 \\
1.05 & 0.022 & 0.090 & 0.935 & 0.089 & 0.067 & 0.110 & 0.957 & 0.094 \\
1.12 & 0.007 & 0.090 & 0.924 & 0.088 & 0.026 & 0.116 & 0.957 & 0.098 \\
\bottomrule
\end{tabular}
}}
    \footnotesize\hspace{.5in}\\
	\raggedright Note: The table shows coverage probability and average length of confidence intervals over $2{,}000$ simulation replications. The sample size is $n=2{,}000$ and the number of bootstrap draws is $2{,}000$. The rows highlighted in yellow and orange correspond to the estimated MSE-optimal bandwidths for the point estimators with $L=0$ and $L=1$, respectively.
\end{table}

\begin{table}[hbtp]\renewcommand{\arraystretch}{1.2}
	\caption{Bootstrap 95\% Confidence Intervals for Partially Linear Regression Model: DGP 3 ($d=3$).}
	\label{Table: plr d=3}
	{\resizebox{\columnwidth}{!}{\begin{tabular}{r cccc cccc}
\toprule
& \multicolumn{4}{c}{$B = 1$} & \multicolumn{4}{c}{$B = 3^{1/d}$} \\
\cmidrule(lr){2-5}\cmidrule(lr){6-9}
& \multicolumn{2}{c}{$L = 0$} & \multicolumn{2}{c}{$L = 1$} & \multicolumn{2}{c}{$L = 0$} & \multicolumn{2}{c}{$L = 1$} \\
\cmidrule(lr){2-3}\cmidrule(lr){4-5}\cmidrule(lr){6-7}\cmidrule(lr){8-9}
$h$ & Coverage & Length & Coverage & Length & Coverage & Length & Coverage & Length \\
\midrule
0.20 & 0.998 & 0.325 & 0.999 & 0.402 & 0.950 & 0.210 & 0.950 & 0.254 \\
0.24 & 0.997 & 0.259 & 0.998 & 0.318 & 0.956 & 0.174 & 0.952 & 0.207 \\
0.28 & 0.993 & 0.217 & 0.996 & 0.263 & 0.957 & 0.152 & 0.961 & 0.177 \\
0.32 & 0.991 & 0.188 & 0.995 & 0.225 & 0.951 & 0.138 & 0.961 & 0.158 \\
0.36 & 0.986 & 0.168 & 0.994 & 0.198 & 0.949 & 0.129 & 0.964 & 0.145 \\
\rowcolor{yellow}  0.40 & 0.979 & 0.154 & 0.994 & 0.179 & 0.938 & 0.122 & 0.963 & 0.135 \\
0.44 & 0.969 & 0.143 & 0.993 & 0.164 & 0.917 & 0.117 & 0.961 & 0.128 \\
0.48 & 0.950 & 0.135 & 0.993 & 0.153 & 0.887 & 0.114 & 0.962 & 0.123 \\
0.50 & 0.934 & 0.131 & 0.992 & 0.149 & 0.873 & 0.113 & 0.966 & 0.121 \\
0.52 & 0.916 & 0.129 & 0.991 & 0.145 & 0.858 & 0.111 & 0.966 & 0.119 \\
0.56 & 0.879 & 0.124 & 0.990 & 0.138 & 0.827 & 0.110 & 0.963 & 0.116 \\
0.60 & 0.830 & 0.120 & 0.988 & 0.132 & 0.778 & 0.109 & 0.961 & 0.114 \\
0.60 & 0.830 & 0.120 & 0.988 & 0.132 & 0.778 & 0.109 & 0.961 & 0.114 \\
0.70 & 0.650 & 0.114 & 0.979 & 0.123 & 0.613 & 0.107 & 0.959 & 0.110 \\
0.80 & 0.408 & 0.110 & 0.975 & 0.117 & 0.399 & 0.108 & 0.958 & 0.107 \\
0.90 & 0.178 & 0.108 & 0.971 & 0.113 & 0.202 & 0.111 & 0.962 & 0.106 \\
\rowcolor{orange}  1.00 & 0.051 & 0.107 & 0.970 & 0.110 & 0.075 & 0.116 & 0.964 & 0.106 \\
1.10 & 0.009 & 0.108 & 0.967 & 0.108 & 0.019 & 0.121 & 0.967 & 0.108 \\
1.20 & 0.000 & 0.109 & 0.960 & 0.107 & 0.002 & 0.128 & 0.969 & 0.110 \\
1.30 & 0.000 & 0.111 & 0.950 & 0.106 & 0.000 & 0.137 & 0.970 & 0.114 \\
1.40 & 0.000 & 0.114 & 0.925 & 0.106 & 0.000 & 0.146 & 0.967 & 0.120 \\
1.50 & 0.000 & 0.118 & 0.887 & 0.107 & 0.000 & 0.156 & 0.955 & 0.126 \\
\bottomrule
\end{tabular}
}}
    \footnotesize\hspace{.5in}\\
	\raggedright Note: The table shows coverage probability and average length of confidence intervals over $2{,}000$ simulation replications. The sample size is $n=2{,}000$ and the number of bootstrap draws is $2{,}000$. The rows highlighted in yellow and orange correspond to the estimated MSE-optimal bandwidths for the point estimators with $L=0$ and $L=1$, respectively.
\end{table}

\begin{table}[hbtp]\renewcommand{\arraystretch}{1.2}
	\caption{Bootstrap 95\% Confidence Intervals for Partially Linear Logit Model: DGP 1 ($d=1$).}
	\label{Table: logit d=1}
	{\resizebox{\columnwidth}{!}{\begin{tabular}{r cccc cccc}
\toprule
& \multicolumn{4}{c}{$B = 1$} & \multicolumn{4}{c}{$B = 3^{1/d}$} \\
\cmidrule(lr){2-5}\cmidrule(lr){6-9}
& \multicolumn{2}{c}{$L = 0$} & \multicolumn{2}{c}{$L = 1$} & \multicolumn{2}{c}{$L = 0$} & \multicolumn{2}{c}{$L = 1$} \\
\cmidrule(lr){2-3}\cmidrule(lr){4-5}\cmidrule(lr){6-7}\cmidrule(lr){8-9}
$h$ & Coverage & Length & Coverage & Length & Coverage & Length & Coverage & Length \\
\midrule
0.10 & 0.952 & 0.319 & 0.953 & 0.324 & 0.944 & 0.308 & 0.947 & 0.311 \\
0.12 & 0.950 & 0.317 & 0.951 & 0.321 & 0.942 & 0.306 & 0.946 & 0.310 \\
0.14 & 0.950 & 0.315 & 0.952 & 0.318 & 0.942 & 0.305 & 0.946 & 0.309 \\
0.16 & 0.951 & 0.314 & 0.953 & 0.317 & 0.944 & 0.304 & 0.945 & 0.308 \\
0.18 & 0.950 & 0.312 & 0.952 & 0.315 & 0.943 & 0.302 & 0.948 & 0.307 \\
0.18 & 0.950 & 0.312 & 0.952 & 0.315 & 0.943 & 0.302 & 0.948 & 0.307 \\
\rowcolor{yellow}  0.20 & 0.950 & 0.311 & 0.951 & 0.314 & 0.941 & 0.301 & 0.947 & 0.306 \\
0.22 & 0.950 & 0.311 & 0.951 & 0.314 & 0.941 & 0.300 & 0.946 & 0.306 \\
0.22 & 0.949 & 0.311 & 0.951 & 0.313 & 0.941 & 0.300 & 0.946 & 0.306 \\
0.24 & 0.948 & 0.310 & 0.950 & 0.313 & 0.940 & 0.299 & 0.945 & 0.305 \\
0.25 & 0.947 & 0.310 & 0.950 & 0.312 & 0.939 & 0.298 & 0.945 & 0.304 \\
0.26 & 0.948 & 0.309 & 0.950 & 0.312 & 0.938 & 0.297 & 0.945 & 0.304 \\
0.28 & 0.948 & 0.309 & 0.949 & 0.312 & 0.935 & 0.296 & 0.943 & 0.303 \\
0.29 & 0.947 & 0.308 & 0.950 & 0.311 & 0.934 & 0.295 & 0.942 & 0.302 \\
0.30 & 0.947 & 0.308 & 0.950 & 0.311 & 0.933 & 0.295 & 0.942 & 0.302 \\
0.32 & 0.946 & 0.307 & 0.949 & 0.311 & 0.933 & 0.293 & 0.941 & 0.300 \\
\rowcolor{orange}  0.36 & 0.942 & 0.306 & 0.948 & 0.310 & 0.931 & 0.291 & 0.940 & 0.298 \\
0.40 & 0.941 & 0.306 & 0.948 & 0.309 & 0.930 & 0.289 & 0.936 & 0.296 \\
0.43 & 0.938 & 0.305 & 0.947 & 0.309 & 0.926 & 0.287 & 0.935 & 0.295 \\
0.47 & 0.939 & 0.304 & 0.948 & 0.308 & 0.922 & 0.285 & 0.935 & 0.293 \\
0.50 & 0.936 & 0.303 & 0.946 & 0.308 & 0.919 & 0.283 & 0.932 & 0.291 \\
0.54 & 0.932 & 0.302 & 0.947 & 0.307 & 0.916 & 0.281 & 0.930 & 0.289 \\
\bottomrule
\end{tabular}}}
    \footnotesize\hspace{.5in}\\
	\raggedright Note: The table shows coverage probability and average length of confidence intervals over $2{,}000$ simulation replications. The sample size is $n=2{,}000$ and the number of bootstrap draws is $2{,}000$. The rows highlighted in yellow and orange correspond to the estimated MSE-optimal bandwidths for the point estimators with $L=0$ and $L=1$, respectively.
\end{table}

\begin{table}[hbtp]
    \renewcommand{\arraystretch}{1.2}
	\caption{Bootstrap 95\% Confidence Intervals for Partially Linear Logit Model: DGP 2 ($d=2$).}
	\label{Table: logit d=2}
	{\resizebox{\columnwidth}{!}{\begin{tabular}{r cccc cccc}
\toprule
& \multicolumn{4}{c}{$B = 1$} & \multicolumn{4}{c}{$B = 3^{1/d}$} \\
\cmidrule(lr){2-5}\cmidrule(lr){6-9}
& \multicolumn{2}{c}{$L = 0$} & \multicolumn{2}{c}{$L = 1$} & \multicolumn{2}{c}{$L = 0$} & \multicolumn{2}{c}{$L = 1$} \\
\cmidrule(lr){2-3}\cmidrule(lr){4-5}\cmidrule(lr){6-7}\cmidrule(lr){8-9}
$h$ & Coverage & Length & Coverage & Length & Coverage & Length & Coverage & Length \\
\midrule
0.14 & 0.982 & 0.507 & 0.985 & 0.568 & 0.952 & 0.404 & 0.952 & 0.431 \\
0.17 & 0.980 & 0.462 & 0.982 & 0.510 & 0.953 & 0.385 & 0.953 & 0.407 \\
0.20 & 0.977 & 0.433 & 0.980 & 0.471 & 0.953 & 0.372 & 0.954 & 0.390 \\
0.22 & 0.974 & 0.414 & 0.978 & 0.445 & 0.951 & 0.364 & 0.951 & 0.380 \\
0.22 & 0.973 & 0.413 & 0.978 & 0.444 & 0.951 & 0.364 & 0.951 & 0.379 \\
0.25 & 0.971 & 0.399 & 0.977 & 0.426 & 0.954 & 0.357 & 0.953 & 0.372 \\
0.27 & 0.969 & 0.392 & 0.975 & 0.416 & 0.950 & 0.354 & 0.953 & 0.367 \\
\rowcolor{yellow}  0.28 & 0.969 & 0.388 & 0.973 & 0.411 & 0.950 & 0.352 & 0.953 & 0.365 \\
0.31 & 0.964 & 0.380 & 0.972 & 0.400 & 0.948 & 0.348 & 0.950 & 0.360 \\
0.32 & 0.963 & 0.378 & 0.970 & 0.398 & 0.948 & 0.347 & 0.951 & 0.359 \\
0.34 & 0.961 & 0.373 & 0.969 & 0.391 & 0.945 & 0.344 & 0.951 & 0.356 \\
0.36 & 0.958 & 0.368 & 0.968 & 0.385 & 0.941 & 0.341 & 0.951 & 0.353 \\
0.36 & 0.958 & 0.368 & 0.967 & 0.385 & 0.940 & 0.341 & 0.951 & 0.353 \\
0.39 & 0.954 & 0.363 & 0.964 & 0.379 & 0.936 & 0.338 & 0.949 & 0.350 \\
0.40 & 0.950 & 0.361 & 0.964 & 0.376 & 0.933 & 0.336 & 0.950 & 0.348 \\
0.42 & 0.946 & 0.359 & 0.964 & 0.374 & 0.928 & 0.335 & 0.951 & 0.347 \\
\rowcolor{orange} 0.45 & 0.938 & 0.356 & 0.963 & 0.370 & 0.919 & 0.332 & 0.948 & 0.344 \\
0.50 & 0.928 & 0.351 & 0.960 & 0.364 & 0.909 & 0.329 & 0.947 & 0.340 \\
0.54 & 0.915 & 0.347 & 0.958 & 0.360 & 0.893 & 0.326 & 0.940 & 0.337 \\
0.59 & 0.898 & 0.344 & 0.953 & 0.356 & 0.876 & 0.323 & 0.938 & 0.334 \\
0.63 & 0.881 & 0.341 & 0.947 & 0.353 & 0.855 & 0.320 & 0.930 & 0.331 \\
0.68 & 0.858 & 0.338 & 0.937 & 0.350 & 0.832 & 0.318 & 0.921 & 0.329 \\
\bottomrule
\end{tabular}

}}
	  \footnotesize\hspace{.5in}\\
    \raggedright Note: The table shows coverage probability and average length of confidence intervals over $2{,}000$ simulation replications. The sample size is $n=2{,}000$ and the number of bootstrap draws is $2{,}000$. The rows highlighted in yellow and orange correspond to the estimated MSE-optimal bandwidths for the point estimators with $L=0$ and $L=1$, respectively.
\end{table}

\begin{table}[hbtp]
    \renewcommand{\arraystretch}{1.2}
	\caption{Bootstrap 95\% Confidence Intervals for Partially Linear Logit Model: DGP 3 ($d=3$).}
	\label{Table: logit d=3}
	{\resizebox{\columnwidth}{!}{\begin{tabular}{r cccc cccc}
\toprule
& \multicolumn{4}{c}{$B = 1$} & \multicolumn{4}{c}{$B = 3^{1/d}$} \\
\cmidrule(lr){2-5}\cmidrule(lr){6-9}
& \multicolumn{2}{c}{$L = 0$} & \multicolumn{2}{c}{$L = 1$} & \multicolumn{2}{c}{$L = 0$} & \multicolumn{2}{c}{$L = 1$} \\
\cmidrule(lr){2-3}\cmidrule(lr){4-5}\cmidrule(lr){6-7}\cmidrule(lr){8-9}
$h$ & Coverage & Length & Coverage & Length & Coverage & Length & Coverage & Length \\
\midrule
0.20 & 0.996 & 1.110 & 0.995 & 1.370 & 0.953 & 0.693 & 0.944 & 0.826 \\
0.23 & 0.994 & 0.863 & 0.994 & 1.049 & 0.951 & 0.578 & 0.949 & 0.673 \\
0.25 & 0.995 & 0.794 & 0.994 & 0.960 & 0.952 & 0.547 & 0.948 & 0.630 \\
0.27 & 0.993 & 0.717 & 0.994 & 0.857 & 0.952 & 0.511 & 0.947 & 0.581 \\
0.30 & 0.991 & 0.647 & 0.994 & 0.765 & 0.952 & 0.479 & 0.951 & 0.538 \\
0.31 & 0.991 & 0.622 & 0.993 & 0.732 & 0.952 & 0.467 & 0.955 & 0.522 \\
0.35 & 0.987 & 0.560 & 0.989 & 0.648 & 0.943 & 0.439 & 0.956 & 0.483 \\
0.35 & 0.987 & 0.559 & 0.989 & 0.646 & 0.943 & 0.438 & 0.956 & 0.482 \\
\rowcolor{yellow} 0.39 & 0.980 & 0.514 & 0.984 & 0.586 & 0.933 & 0.418 & 0.952 & 0.454 \\
0.40 & 0.979 & 0.505 & 0.983 & 0.573 & 0.929 & 0.414 & 0.949 & 0.448 \\
0.43 & 0.971 & 0.482 & 0.983 & 0.541 & 0.926 & 0.403 & 0.951 & 0.434 \\
0.45 & 0.963 & 0.467 & 0.980 & 0.522 & 0.914 & 0.397 & 0.947 & 0.425 \\
0.47 & 0.957 & 0.457 & 0.979 & 0.507 & 0.907 & 0.392 & 0.948 & 0.418 \\
\rowcolor{orange} 0.50 & 0.944 & 0.441 & 0.976 & 0.486 & 0.891 & 0.384 & 0.945 & 0.408 \\
0.51 & 0.939 & 0.438 & 0.976 & 0.482 & 0.887 & 0.383 & 0.942 & 0.406 \\
0.55 & 0.917 & 0.423 & 0.972 & 0.461 & 0.859 & 0.375 & 0.939 & 0.397 \\
0.55 & 0.914 & 0.422 & 0.971 & 0.460 & 0.856 & 0.375 & 0.937 & 0.396 \\
0.58 & 0.891 & 0.411 & 0.964 & 0.445 & 0.832 & 0.369 & 0.928 & 0.389 \\
0.60 & 0.875 & 0.407 & 0.960 & 0.440 & 0.814 & 0.367 & 0.927 & 0.387 \\
0.65 & 0.829 & 0.396 & 0.945 & 0.425 & 0.774 & 0.362 & 0.908 & 0.379 \\
0.70 & 0.774 & 0.387 & 0.926 & 0.413 & 0.727 & 0.357 & 0.888 & 0.373 \\
0.75 & 0.720 & 0.380 & 0.902 & 0.403 & 0.663 & 0.353 & 0.855 & 0.368 \\
\bottomrule
\end{tabular}}}
    \footnotesize\hspace{.5in}\\
    \raggedright Note: The table shows coverage probability and average length of confidence intervals over $2{,}000$ simulation replications. The sample size is $n=2{,}000$ and the number of bootstrap draws is $2{,}000$. The rows highlighted in yellow and orange correspond to the estimated MSE-optimal bandwidths for the point estimators with $L=0$ and $L=1$, respectively.
\end{table}

\begin{figure}[htbp]
  \centering
\caption{Coverage probabilities as a function of bandwidth $h$.}
  \label{fig:coverage}

  % ── d=1 ────────────────────────────────────────────────────────────────
  \begin{subfigure}[t]{0.48\textwidth}
    \includegraphics[width=\linewidth]{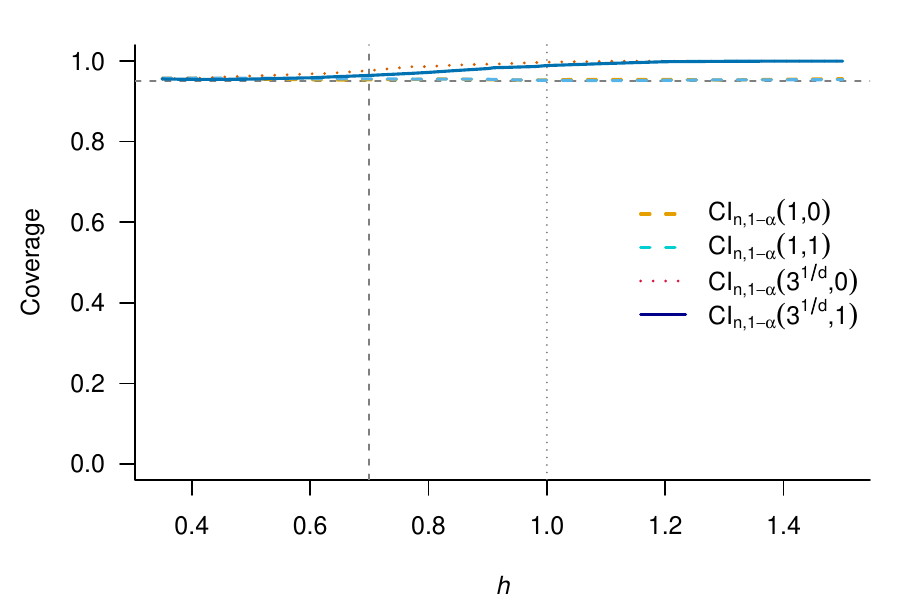}
    \caption{PLR ($d=1$).}
  \end{subfigure}
  \hfill
  \begin{subfigure}[t]{0.48\textwidth}
    \includegraphics[width=\linewidth]{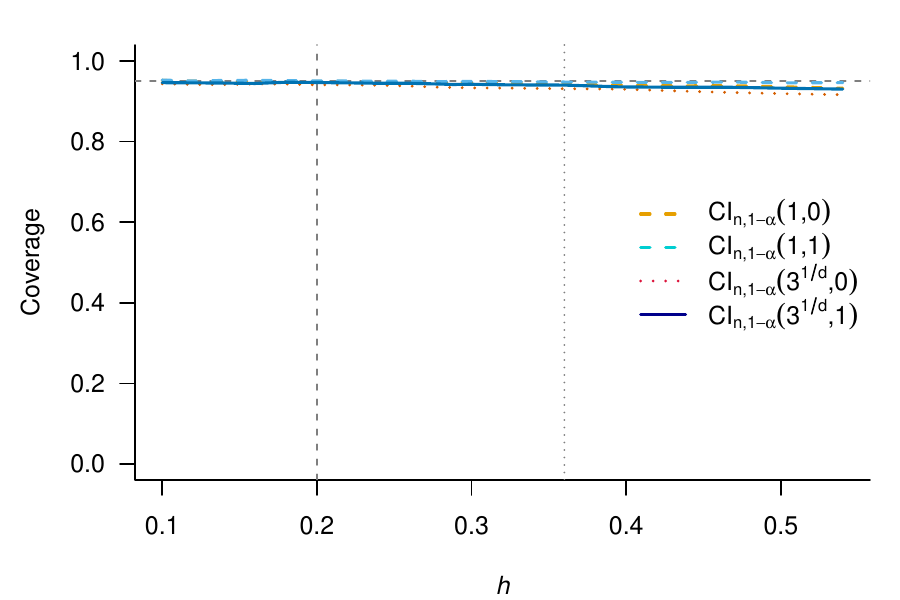}
    \caption{PLL ($d=1$).}
  \end{subfigure}

  \medskip

  % ── d=2 ────────────────────────────────────────────────────────────────
  \begin{subfigure}[t]{0.48\textwidth}
    \includegraphics[width=\linewidth]{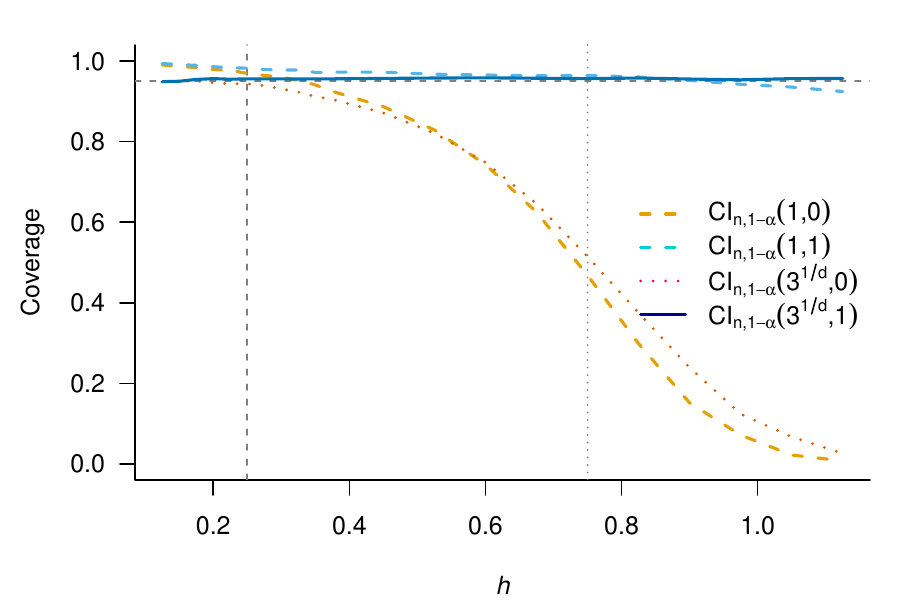}
    \caption{PLR ($d=2$).}
  \end{subfigure}
  \hfill
  \begin{subfigure}[t]{0.48\textwidth}
    \includegraphics[width=\linewidth]{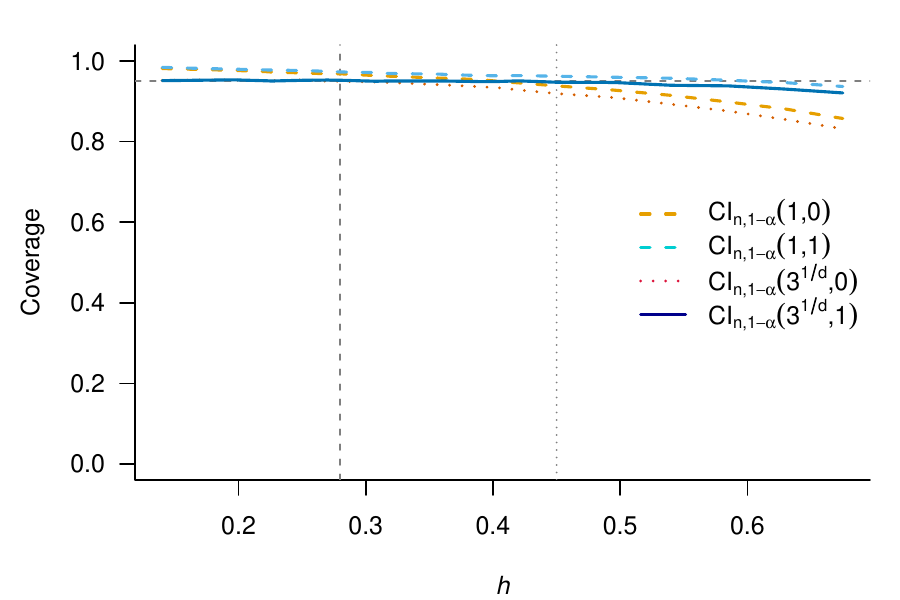}
    \caption{PLL ($d=2$).}
  \end{subfigure}

  \medskip

  % ── d=3 ────────────────────────────────────────────────────────────────
  \begin{subfigure}[t]{0.48\textwidth}
    \includegraphics[width=\linewidth]{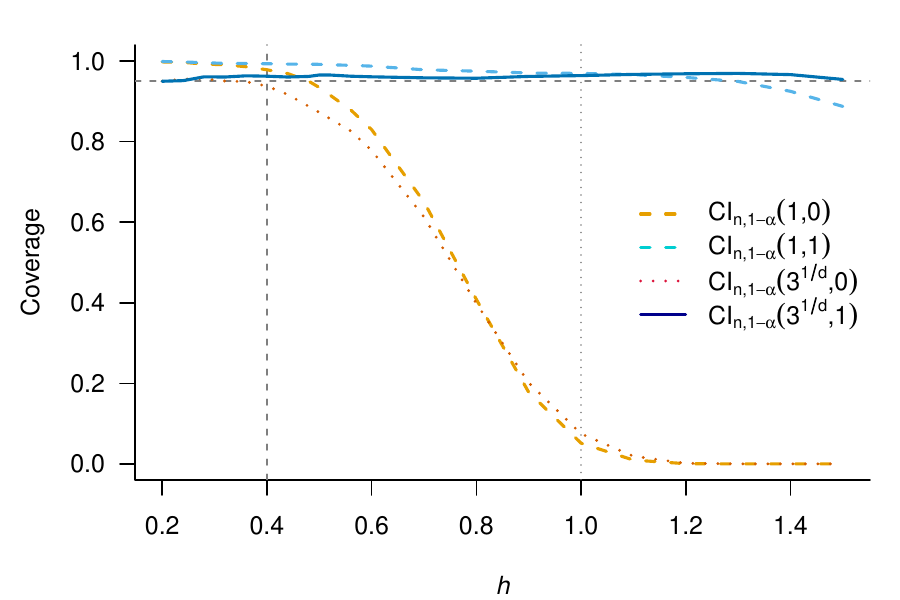}
    \caption{PLR ($d=3$).}
  \end{subfigure}
  \hfill
  \begin{subfigure}[t]{0.48\textwidth}
    \includegraphics[width=\linewidth]{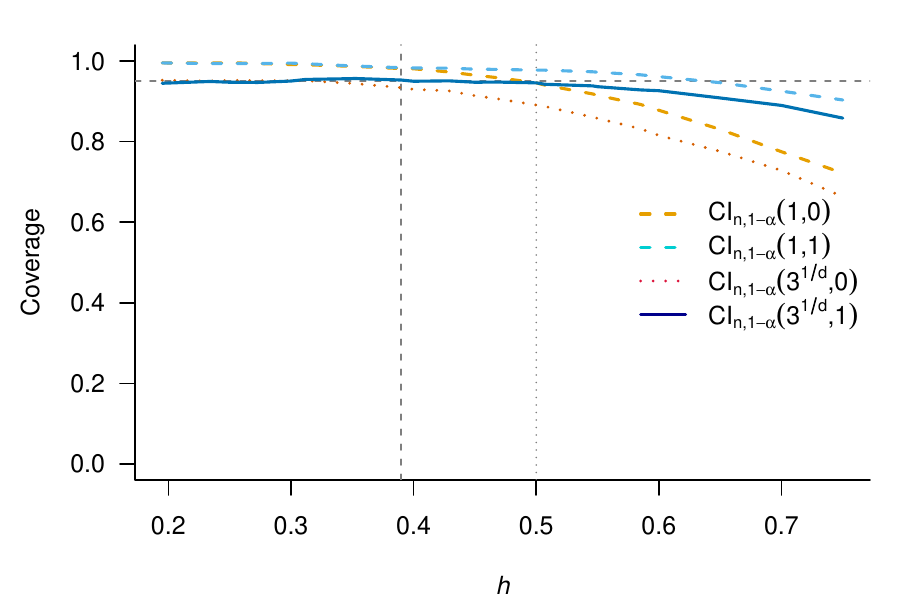}
    \caption{PLL ($d=3$).}
  \end{subfigure}
  \footnotesize\hspace{.5in}\\
  \raggedright Note: Each panel plots the empirical coverage probability of $\mathsf{CI}_{n,0.95}^*(B,L)$ against bandwidth $h$ for the four procedures $(B,L)\in\{(1,0),(1,1),(3^{1/d},0),(3^{1/d},1)\}$. The horizontal dotted line marks the $95\%$ nominal level. The two vertical lines mark the estimated MSE-optimal bandwidths $h_{L=0}$ and $h_{L=1}$ for the point estimators $\widehat{\btheta}_n$ ($L=0$) and $\widetilde{\btheta}_n$ ($L=1$), respectively. The sample size is $n=2{,}000$ and results are based on $2{,}000$ simulation replications with $2{,}000$ bootstrap samples.
\end{figure}

\renewcommand{\arraystretch}{1.2}

\bibliography{CJN_2026_ET--bib}
\bibliographystyle{ecta}

\end{document}